\documentclass{LMCS}

\def\dOi{10(3:23)2014}
\lmcsheading%
{\dOi}
{1--30}
{}
{}
{Mar.~26, 2013}
{Sep.~18, 2014}
{}

\ACMCCS{[{\bf Computing methodologies}]: Artificial
  intelligence---Knowledge representation and reasoning;[{\bf Theory
      of computation}]: Formal languages and automata theory; Design
  and analysis of algorithms---Graph algorithms analysis }


\usepackage{graphicx}
\usepackage[cp1250]{inputenc}
\usepackage{amsmath}
\usepackage{amssymb}
\usepackage{color,graphics}
\usepackage{t1enc}
\usepackage{hyperref}
\usepackage{xspace}
\usepackage{tikz}
\usetikzlibrary{shapes,mindmap,trees,decorations.pathreplacing,decorations.pathmorphing,shapes,matrix}
\usetikzlibrary{positioning}
\usetikzlibrary{fit}
\usetikzlibrary{calc}
\usetikzlibrary{decorations.pathmorphing}
\usetikzlibrary{decorations.markings}

\newcommand{\coloneq}[0]{\mathrel{\mathop:}=}

\def\nopar{{\parfillskip0pt\parskip0pt\par\noindent}}
\newcommand{\TS}[0]{\mathit{TS}}
\newcommand{\pppp}[1]{\underset{#1}{\curvearrowright}}
\newcommand{\RRR}[1]{{\mathfrak{R}(#1)}}
\newcommand{\DAG}{{\mbox{DAG}}}
\newcommand{\AHP}{{\mbox{AHP}}}
\newcommand{\DAGcut}[2]{{{#1}_{|#2}}}
\newcommand{\tauE}[2]{{{#1}_{#2}}}
\newcommand{\pipart}[3]{{{#1}_{#2}^{#3}}}
\newcommand{\vertexOF}[1]{{\mathcal{V}(#1)}}
\newcommand{\pidigraph}[2]{{\mathcal{G}(#1,#2)}}
\newcommand{\FAS}{\mbox{FAS}}
\newcommand{\VS}{\mbox{VC}}
\newcommand{\pipath}[2]{{\mathfrak{h}^{#1}(#2)}}
\newcommand{\gadget}{\mathcal{N}}
\newcommand{\gadgetV}{\mathcal{V}}
\newcommand{\gadgetE}{\mathcal{E}}
\newcommand{\G}{\mathcal{G}}
\newcommand{\V}{\mathbb{V}}
\newcommand{\E}{\mathbb{E}}
\newcommand{\VV}{{\mathbb{V}_\G}}
\newcommand{\EE}{{\mathbb{E}_\G}}
\newcommand{\F}{{\mathcal{F}}}
\newcommand{\DD}[2]{\mathcal{P}(#1,#2)}
\newcommand{\e}{{e}}
\newcommand{\ee}[1]{\e(\!#1\!)}
\newcommand{\eT}[1]{\e^{-1}\!(#1)}
\newcommand{\LL}{{{\boldsymbol{\swarrow}}}}
\newcommand{\RR}{{{\boldsymbol{\searrow}}}}
\newcommand{\miz}[1]{\texttt{#1}}\newcommand{\C}[0]{$^{-1}$}\newcommand{\M}[0]{$\cdot$}\newcommand{\QQ}[1]{#1:\:}

\tikzset{
    Referencja/.style={-latex,draw=black,thick},
    Variable/.style={-latex,draw=black,thick,dashed},
    mynodes/.style={circle, draw, minimum size=1mm,
         draw=black},
    myrectanglenodes/.style={
    rectangle,rounded corners, draw, minimum size=5mm,
         draw=black},
    mycirclenodes/.style={
    circle, draw, minimum size=8mm,
         draw=black}
         }
\def\s{0.4cm}
\def\p{0.125cm}
\def\g{0.77cm}

\newcommand{\picut}[2]{{{#1}_{|#2}}}
\newcommand{\pipathcut}[3]{{\mathfrak{h}^{#1}_{|#2}(#3)}}

\begin{document}

\title[Improving legibility of natural deduction proofs is not trivial]{Improving legibility of natural deduction proofs\\is
not trivial}
\author{Karol P\k{a}k}

\email{pakkarol@uwb.edu.pl}
\address{University of Bialystok, Institute of Computer Science, Bialystok, Poland}
\thanks{The paper has been financed by the resources of the Polish National Science Centre granted by decision no
DEC-2012/07/N/ST6/02147}

\begin{abstract}
In formal proof checking environments such as Mizar it is not merely the validity of mathematical formulas that is evaluated in
the process of adoption to the body of accepted formalizations, but also the readability of the proofs that witness validity.
As in case of
  computer programs, such proof scripts may sometimes be more and sometimes be less readable. To better understand the notion
  of readability of formal proofs, and to assess and improve their readability, we propose in this paper a method of improving
  proof readability based on Behaghel's First Law of sentence structure.  Our method maximizes the number of local references
  to the directly preceding statement in a proof linearisation.  It is shown that our optimization method is NP-complete.
 \end{abstract}


\maketitle

\section{Introduction}

\subsection{Motivations}

The readability of formal proofs has significantly influenced
 development of formalization. The idea of formalization was recognized by David Hilbert
 in his program of formalization of mathematics\,---\,research in
 his program led, in particular, to the concept of the universal
 Turing machine.
Hilbert's project in his most idealistic form collapsed due to G$\ddot{\mbox{o}}$del's incompleteness theorems. In spite of
them, a group of mathematicians under a single pen name Nicolas Bourbaki tried to
 formalize an exposition of modern advanced mathematics using
 first-order logic.
 This project was also abandoned because the resulting formal proofs turned out to be too obscure and
the formalization in such a detailed form started to be perceived as hardly readable \cite{Snapper} and as a result
impractical.
 Moreover,
checking the correctness of formal proofs turned out to be as difficult as the reading comprehension of these proofs.
 Not only Bourbaki's project
was faced with the readability problems. The earlier efforts of Whitehead and Russell in the formalization of mathematics had
the same problems. The formal definitions and proofs that they created were deemed cryptic as the approach to mathematics they
used was extremely rigorous and precise. The formalization has become feasible only after increasing availability of computers.
These problems are described in more detail by Zammit \cite{PhdZammit}.

Proof development in formalized mathematical frameworks is similar to program development. In both cases the content is created
in an artificial language that has precise meaning executed by a computer system. Therefore it comes as no surprise that both
activities have much in common. Since readability of programs has a significant impact on their maintainability
\cite{Aggarwal02} it is strongly expected that readability of proof scripts has a significant impact on proof maintainability.
Indeed the maintenance tasks take place in current formal proof development, e.g.\ novice users of the formal proof development
tools often follow the existing formalizations, and the formalizations are subject to refinement for instance when a
generalized version of a theorem is to be proved. This concerns especially systems such as Isabelle/Isar \cite{Isar11} or
Mizar
\cite{Naumowicz:2009:BOM} where the proof script language is close to the natural
language.

In this paper, we focus on another, still underdeveloped approach. Models of cognitive perception of read material stress that
locality of reference is a significant factor in the process of understanding. This can be summarized with the Behaghel's First
Law that can be stated as follows, {\em elements that belong close together intellectually will also be placed close together}
\cite{Behaghel1909}. This law of local reference is also recognized in
{modern scientific literature concerning human perception} \cite{Gibson98,RLevy}. The significance of the locality of reference
is exploited in the current study. We focus on a particular proof step in a proof script. The step uses some information that
was derived before in the proof. This information can be located somewhere far away in the proof or within a close neighborhood
of the step. With Behaghel's law in mind, we assume that a step where at least part of
 information it requires is available in the directly preceding
step is more comprehensive than a step in which all information is far away in the proof {(as a significant part of assumptions
can still be in the, so called, {\em working memory}, see e.g.\ \cite[Chapter~1]{Cowan98}). }

\par The locality of reference principle is the basis for a method of proof
distillation that is presented in this paper. Since the steps that refer to the preceding step are perceived as more
comprehensive, the procedure transforms the proof scripts so that the number of steps that violate the requirement is minimal.
The transformation consists of two actions.  The first one translates a linear representation of a proof script to a
representation in the form of a directed acyclic graph and the second one chooses a different linearisation of the proof so
that the number of proof steps that are,
 according to the presented principles, difficult to comprehend
 is minimal.

\subsection{Related works}

The experience of big proof formalization developments shows that formalized results are often used as libraries, which requires
reading proof scripts \cite{Schwarzweller2009}. But there are many more reasons that force formalization authors to read existing proof scripts.
New users often use the available library as an in-depth tutorial \cite{KPslgr}. Experienced users adapt or modify existing proofs to obtain
stronger theorems \cite{Gontie}
 (actually this is mentioned in page 3 of an unpublished
preliminary version of the article \cite{UnpublishedGonther}).
 
 Authors that cooperate on a common formalization analyze each other's achievements
(e.g. over 40\% of MML proof scripts have more that one author). Therefore it comes as no surprise that formal systems develop in
many directions that improve their readability.

One of the most popular methods to make it possible for computers to assist in formal reasoning is based on bringing
the formal mathematical language to the informal one by introduction to the formal language idioms that stem from informal
mathematical practice. In particular, the keywords
\miz{thus} and
\miz{show}(\miz{s}) indicate a conclusion of reasoning in Mizar and
Isabelle/Isar systems respectively (for more details see~\cite{MIZAR,Isar11}). In both systems the same
\miz{assume}(\miz{s}) keyword is used to indicate a step which,
e.g., assumes the antecedent of a thesis that is an implication. Such constructions, which accord to the rules of natural
deduction created by S. Ja\'skowski and F. B. Fitch \cite{FBF,JS} were introduced in the Mizar system and later
implemented in a number of procedural systems: {\em Declare} \cite{Sym1999}, {\em Mizar Mode for HOL} \cite{Harrison1996}, {\em Mizar-light
for HOL-light}
\cite{MLHL}, {\em \verb*Vmiz3V for HOL-light} \cite{MIZ3}, {\em
MMode for Coq} \cite{Giero}, {\em declarative proof language (DPL) for Coq} \cite{Corbineau2007}. The process of making the
Mizar language similar to the informal one is based also on introduction of informal constructions such as ellipsis
\cite{AK_ellipsisCICM2012}. In particular, in the MML database we
can find theorems formulated as follows:
\miz{k <= n implies k = 0 or ... or k = n}, and this simply states that
for every natural numbers $k$ and $n$, if $k\leq n$, then $k$ can be equal to $1$ or $2$ or $\ldots$ or $n$.

Mizar developers are generally very careful when adding new syntactic constructs to the language, and always prefer
symbols that stem from standard English. All logical symbols in Mizar, as well as symbols denoting functors,
predicates and attributes in the Mizar language are typed as plain text in the editor
(the editor of choice for a majority of Mizar users is MizarMode for Emacs created by Urban \cite{Urban05-jal}).
Conversion of text symbols to \LaTeX{} is performed only when the proof script is published in the dedicated journal ``Formalized Mathematics''.
In this respect, the
editor used for Isabelle/Isar\,---\,jEdit \cite{JEdit} is more flexible and can 
visualize proofs both as plain text or in \LaTeX{} symbols.
Clearly these editors provide many features typical for programming language editors such as syntax highlighting or hint
systems (infotip). Additionally, MizarMode is connected
 with the {\em Automated Reasoning for Mizar} (\mbox{Miz}$\mathbb{AR}$) service
 that provides ATP assistance to Mizar authors in searching a proof of a given statement \cite{ATP}
 and also
facilitates access to the tools distributed with the Mizar system, which
can to some extent improve the proof scripts legibility.
This initial improvement generally is based on finding and removing irrelevant parts of reasoning or modifying the
justification of steps to simplify the list of used premises, preserving the correctness of the modified proof scripts. In
particular, the utility, called {\em Irrelevant Premises Detector} (distributed under the name \mbox{RELPREM}), indicates
premises which are not needed to justify a given step in reasoning (for more details see~\cite{MIZAR}). Another, equally popular
utility, {\em{}Irrelevant Inferences Detector} (\mbox{RELINFER}), is often used by authors
when they complete writing a Mizar article. This tool indicates every reference to an irrelevant step, i.e.\ a reference which
can be replaced by all references used to justify this step without affecting the correctness of reasoning. In consequence
such modification can shorten proof scripts, but also can spoil their readability. Indeed, steps irrelevant for the machine can be
crucial for human understanding of the proof idea. It is common that inferences which are obvious for proof writers are not
obvious for the checker and vice versa \cite{Rudnicki87}. Moreover, a simple variant of resolution that is implemented in the Mizar system does not
approximate well these two kinds of obviousness. 
As a result the Mizar developers distribute only tools that enable
finding such references, but the decision about their removal is left
to the proof author to make it possible to preserve proof legibility.

Independent efforts of Urban concerning visualization of existing Mizar proof scripts led to creation of an auxiliary tool that
converts every script to a linked HTML form
(\cite{UrbanXML}, see also http://mizar.org/version/current/html/).
This linked form includes additional information that is not explicitly available in the script, but is
calculated by the Mizar analyzer during the verification process. In particular, it gives immediate access to every definition
that has been used implicitly in reasoning and every definition of functor, predicate, and attribute that appears in the
statement of a proof step. It is worth to note that the MML script authors often use the same symbol to express different
things, e.g.,
\miz{*} and \miz{+} are defined or redefined for different notions in MML more than 130 and 160 times, respectively.
Therefore, it comes as no surprise that in scripts with a more complex environment, even the authors have a problem in
recognizing the symbol variant they used.

Another direction of improving the legibility of Mizar proof scripts is the automation of the most frequently performed
operations in reasoning, which was explored by Naumowicz and Byli\'nski
\cite{Naumowicz:2010:IEC}. This approach is based on the fact that
premises that are simply a substitution of some functor (e.g., commutativity, involutivnes) or predicate (e.g., reflexivity,
symmetry) do not have to be formulated explicitly in reasoning if these properties have been already proven and
substituted terms occur in the reasoning. A special kind of information that is automatically generated and used in justifications
of the Mizar proof steps, called {\em reduction}, has been implemented by Korni{\l}owicz \cite{Kornilowicz13}. Here, a reduction
means an equality $f(t_1,t_2,\ldots,t_n)=t$ where $f$ is an $n$-ary function symbol, $t_1,t_2,\ldots, t_n$ are terms,
and $t$ is a subterm of $f(t_1,t_2,\ldots,t_n)$, e.g., the restriction of a function $f$ to its domain is equal to $f$.

An orthogonal approach to improving the legibility of formal proof scripts is based on popular practices that solve the problem
of illegible long reasoning in informal mathematical practice. This approach is generally based on two methods. The first one
consists in finding (often less important) fragments of reasoning and then extracting them in the form of a lemma, or encapsulating them
at deeper levels of nested proofs. The second ones tries to reorganize the order of independent steps
focusing mainly on premises used to justify proof steps. The first experiments with the first method were carried out generally
to reduce the size of mechanically generated proofs \cite{LEM_EX}. It is generally known that theorem provers often justify the
same goal and in exactly the same way several times. Therefore, it is not surprising that the effort to extract similar proofs as
lemmas is often undertaken. However, such modifications are generally carried out without the emphasis on the legibility of obtained
scripts. Only the initial studies on properties of reasoning passages that in the reader's opinion determine fragments that
deserve to be called a lemma were conducted \cite{JAR1}. Additionally, there were considered dependencies between such
properties and the complication level of statements that describe reasoning in selected passages.

Experiential tools that try to realize the second approach were created by P\k{a}k for the Mizar system \cite{KPslgr}. The
impact of these tools, which are based on greedy algorithms, has been already recognized by the community of Mizar authors \cite{JAR1}.
Additionally, these techniques of proof script modification are taken into consideration in
other systems as an effective way to refactor machine proofs
\cite{Redirecting2013,KaliszykUrban,Robust2013}.
However, techniques that try to optimise proof readability criteria in most cases correspond to problems for which
complexity is still an open problem.

\medskip

In this paper we study the computational complexity of methods that have origins in Behaghel's law. For this we devise a
generalization of two known NP-complete optimization problems that have been already described in \cite{NPLIST} and reduce to
these problems two proof readability criteria. This shows that the task of improving readability is computationally demanding
and so non-trivial.

\par
In Section~\ref{Graph_representation_of_proofs} we introduce the notion of an abstract model of natural deduction proofs and we
discuss selected methods to improve readability. In Section~\ref{Formulation_of_the_Behaghel's_Law_Problem} we restrict our
attention to two problems of improving the legibility of natural deduction proofs that have origins in Behaghel's law. Then in
Section~\ref{Acyclic_Hamiltonian_Partition} we formalize these problems in terms of the acyclic partitions and in
Section~\ref{NPproof} we show that these problems are NP-complete, since Minimum Feedback Arc Set problem can be {transformed}
to {them}. In Section~\ref{Constructive_abstract_proof_graphs} we present {a} family of abstract proof graphs for which there
{exist} proofs written in the Mizar system that have structures described by these graphs. Section~\ref{Other} contains a brief
overview of other problems {associated with} improving proof legibility and the study of their complexity. Finally,
Section~\ref{Conclusion} concludes the paper and discusses the future work.

\section{Graph representation of proofs}\label{Graph_representation_of_proofs}

The abstract notion of proof we deal with here is based on natural deduction created by S. Ja\'skowski and F. B. Fitch
\cite{FBF,JS,MW}.
As an illustration, let us consider an example written in this notation, which is presented in Fig.~\ref{some_proof}.
\begin{figure}[h!]
\begin{center}\footnotesize
\begin{tabular}{cccllc}
1&$\strut\mkern-18mu$&&{$\exists_x P(x)$} & premise&$\strut\mkern-18mu$\\ 2&$\strut\mkern-20mu$&&{$\forall_x \left(P(x)
\rightarrow Q(x)\right)$} & premise&$\strut\mkern-18mu$\\ 3&$\strut\mkern-18mu$&&{$\forall_x\forall_y\left(Q(x) \rightarrow
R(y)\right)$} & premise&$\strut\mkern-18mu$\\\cline{2-6} 4 &
\multicolumn{1}{|c}{$\strut\mkern-18mu$}&\multicolumn{1}{l}{$\strut\:y_0\:$} && &\multicolumn{1}{l|}{$\strut\mkern-18mu$}\\
\cline{3-5}
5& \multicolumn{1}{|c}{$\strut\mkern-18mu$}&\multicolumn{1}{|l}{$\strut\:x_0\:$} & {$P(x_0)$} & \multicolumn{1}{l|}{assumption}
& \multicolumn{1}{l|}{$\strut\mkern-18mu$}\\ 6& \multicolumn{1}{|c}{$\strut\mkern-18mu$}&\multicolumn{1}{|l}{} &
{$P(x_0)\rightarrow Q(x_0)$} & \multicolumn{1}{l|}{$\forall_x$e 2} & \multicolumn{1}{l|}{$\strut\mkern-18mu$}\\ 7&
\multicolumn{1}{|c}{$\strut\mkern-18mu$}&\multicolumn{1}{|l}{} & {$\forall_y\left(Q(x_0)\rightarrow R(y)\right)$} &
\multicolumn{1}{l|}{$\forall_x$e 3} & \multicolumn{1}{l|}{$\strut\mkern-18mu$}\\ 8&
\multicolumn{1}{|c}{$\strut\mkern-18mu$}&\multicolumn{1}{|l}{} & {$Q(x_0)\rightarrow R(y_0)$} &
\multicolumn{1}{l|}{$\forall_y$e 7} & \multicolumn{1}{l|}{$\strut\mkern-18mu$}\\ 9&
\multicolumn{1}{|c}{$\strut\mkern-18mu$}&\multicolumn{1}{|l}{} & {$Q(x_0)$} & \multicolumn{1}{l|}{$\rightarrow$e 6,5} &
\multicolumn{1}{l|}{$\strut\mkern-18mu$}\\ 10& \multicolumn{1}{|c}{$\strut\mkern-18mu$}&\multicolumn{1}{|l}{} & {$R(y_0)$} &
\multicolumn{1}{l|}{$\rightarrow$e 8,9} & \multicolumn{1}{l|}{$\strut\mkern-18mu$}\\
\cline{3-5}
11 & \multicolumn{1}{|c}{$\strut\mkern-18mu$}&\multicolumn{1}{l}{} &{$\strut R(y_0)$} & \multicolumn{1}{l}{$\exists_x$e
1,5--10}&
\multicolumn{1}{l|}{$\strut\mkern-18mu$}\\\cline{2-6}
12&\multicolumn{1}{c}{$\strut\mkern-18mu$}&\multicolumn{1}{l}{} &{$\strut\forall_y R(y)$} & \multicolumn{1}{l}{$\forall_y$i
4--11}&
\multicolumn{1}{l}{$\strut\mkern-18mu$}\\
\end{tabular}
\end{center}
\caption{The example of reasoning that is based on the Fitch notation of natural deduction. It justifies the formula
$\left(
\left(\exists_x P(x)\right)
\wedge
\left(\forall_x(P(x)\!\rightarrow\! Q(x))\right)\right.$ $\left.\wedge
\left(\forall_x\forall_y(Q(x)\!\rightarrow\! R(y))\right)\right)
\rightarrow \forall_y R(y)$.}
\label{some_proof}
\end{figure}
In the same way as M. Hurt and M. Ryan~\cite{LinCS}, we give a unique number for every line (written on the left), and in the
justification we write references to all lines that are used. Additionally, we denote by $\rightarrow\!\mbox{e}$ the Modus
Ponens rule, $\frac{\phi\:\:\phi\rightarrow\psi}{\psi}$; and by $\forall_x\mkern1mu\mbox{e}$ the rule for eliminating the
universal quantifier $\forall$, $\frac{\forall_{x}\phi}{\phi[t/x]}$. Rules for eliminating $\exists$ and introducing $\forall$
use boxes to stipulate the scope of an auxiliary variable $x_0$ as follows:
\begin{equation}
\begin{array}{clcl}
   \begin{array}{c}
       \begin{array}{|lc|}\hline x_0 & \\ &\vdots\\&\phi[x_0/x]\\\hline
       \end{array}
        \vspace{1mm}\\\hline
    \end{array}\:\vbox{\hbox{,}\vspace{-22pt}}& &
    \begin{array}{c}
      \begin{array}{c|lc|}\cline{2-3}& x_0\: & \phi[x/x_0]\\& & \vdots  \\
            \exists_x\phi\;
            & &\chi\\\cline{2-3}
      \end{array}
      \vspace{1mm}\\\hline
    \end{array}\:\vbox{\hbox{,}\vspace{-22pt}}&\\
\forall_x\phi&\mkern20mu&\chi&
\end{array}
\end{equation}
and we denote them by $\forall_x\mbox{i}$, $\exists_x\mbox{e}$ respectively. The variables $x_0$ are called eigenvariables.
This means that in case of $\forall_x\phi$ it cannot occur in the assumptions of the derivation concluding with $\phi[x_0/x]$
and in case of $\chi$ it cannot occur in $\chi$.

\par An abstract model of such proofs was considered in
\cite{KPslgr} for the general case, but for our purposes we consider a
simplified model which is represented by a directed acyclic graph (\DAG{}). A \DAG{} $G=\langle V,E\rangle$ with a
distinguished set of arcs $\RRR{G}\subseteq E$\label{RRRRR}
 is called an
abstract proof graph. The vertices of $G$ represent steps of reasoning and arcs of $G$ represent the flow of information
between different steps of reasoning. Additionally, an arc of $\RRR{G}$, called a {\em reference arc}, describes the dependence
between the expression of a step (the head of the arc) and a previously justified step (the tail of the arc) used as the
justification of that expression. Other arcs of $G$ describe all kinds of additional constraints
   that force one step to precede another step,
e.g.\ the dependence between steps which introduce variables into the reasoning and steps which use these variables in the
expressions.

In general, the reasonings in Ja\'skowski-Fitch notation can have nested subreasonings, i.e.\ ones that are enclosed within
boxes created by quantifier rules. The graph model presented here does not take into account this structure. However, the
NP-hardness proofs presented in this paper also apply to the appropriate generalization, which includes the nested structure,
as they are its special case.

\medskip

\par To illustrate the model,
let us consider the proof contained in the scope box of $x_0$ (Fig.~\ref{some_proof}). Its model is presented in Fig.~\ref{APG}
where solid arrows represent reference arcs, dashed arrows represent the uses of variable $x_0$, and the arrows both solid and
dashed represent both reference and the use of a variable. Additionally, both arcs and nodes of the actual abstract proof graph
are not labeled (arcs and nodes in Fig.~\ref{APG}, are labeled only to simplify their identification).
\begin{figure}[ht!]
\begin{center}
\begin{tikzpicture}
\footnotesize
     [>=stealth,->,shorten >=2pt,looseness=.5,auto]
    \matrix (M)[matrix of math nodes,
        column sep={3.0cm,between origins},
        row sep={1.4cm,between origins},
        nodes=myrectanglenodes]
        {
\node(A6){P(x_0)\!\rightarrow\!Q(x_0)};&
\node(A5){x_0\:\:P(x_0)};&
\node(A7){\forall_y\left(Q(x_0)\!\rightarrow \!R(y)\right)};\\
\node(A9){Q(x_0)};&\node(A10){R(y_0)};&\node(A8){Q(x_0)\!\rightarrow\!R(y_0)};\\     };
\draw[Variable](A5)--(A6);
\draw[Referencja]($(A5.south west)+(0.02,-0.01)$)--($(A9.north east)+(0.02,-0.01)$);
\draw[Variable]($(A5.south west)+(-0.02,0.01)$)--($(A9.north east)+(-0.02,0.01)$);
\draw[Referencja](A6)--(A9);
\draw[Referencja](A7)--(A8);
\draw[Referencja](A8)--(A10);
\draw[Referencja](A9)--(A10);
\draw[Variable](A5)--(A7);
\draw[Variable](A5)--(A8);
\node at ($0.5*(A5.west)+0.5*(A6.east)+(0,0.2)$){$x_0$};
\node at ($0.5*(A5.east)+0.5*(A7.west)+(0,0.2)$){$x_0$};
\node at ($0.5*(A5.south east)+0.5*(A8.north west)+(0.2,0.2)$){$x_0$};
\node at ($0.5*(A5.west)+0.5*(A9.east)+(0,0.2)$){$x_0$};
\node at ($0.5*(A6.south)+0.5*(A9.north)+(-0.3,0)$){$\rightarrow$e};
\node at ($0.5*(A9.east)+0.5*(A10.west)+(0,0.12)$){$\rightarrow$e};
\node at ($0.5*(A8.west)+0.5*(A10.east)+(0,0.12)$){$\rightarrow$e};
\node at ($0.5*(A5.west)+0.5*(A9.east)+(0.2,-0.15)$){$\rightarrow$e};
\node at ($0.5*(A7.south)+0.5*(A8.north)+(0.29,0)$){$\forall_y$e};
\end{tikzpicture}
\end{center}
\caption{The abstract proof graph illustrating
the proof contained in the scope box of $x_0$ in Fig.~\ref{some_proof}.}
\label{APG}
\end{figure}

\par We identify here topological sortings of $G$ with one-to-one functions\break $\tau:V\longrightarrow \{1,2,\ldots,|V|\}$
such that, for every arc $\langle u,v\rangle\in E$, $\tau(u)<\tau(v)$.
Such sortings are sometimes called {\it linearizations}. We denote by $\TS(G)$ the set of all topological sortings of $G$. We
call a vertex of $V$ {\em ``then'' $\tau$-step} if a step of reasoning which corresponds to this vertex uses in justification
the immediately preceding step of the derivation induced by $\tau$ (and possibly some other statements). To be more precise,
vertex $u$ is a ``then'' $\tau$-step if and only if for some vertex $v$, $\tau(v)+1=\tau(u)$ and $\langle v,u\rangle\in
\RRR{G}$. For example, the step ``$10\;\;R(y_0)\;\;\rightarrow\mbox{e}\,8,9$'' (Fig.~\ref{some_proof}) is a ``then'' step.

\par To make references to the directly preceding step in a reasoning explicit, we extend
the Fitch notation with a new kind of step, called ``then'' step. Suppose that a step $s$ in line $n\!+\!1$ justifies a formula
$\phi$ by a rule $\mathfrak{r}$ that refers to the line $n$ and is represented as
\begin{equation}
\begin{array}{lcr}
n\!+\!1 &\phi&\mathfrak{r}\, \ldots,n,\ldots\,,
\end{array}
\end{equation}
We present it as a ``then'' step of the form
\begin{equation}
\begin{array}{lcr}
n\!+\!1 &then\:\:\phi&\mathfrak{r}\, \ldots,\ast,\ldots\,,
\end{array}
\end{equation}
where the number of the previous line n is replaced by $\ast$ (see Fig.~\ref{reorganisation}).
\begin{figure}[h!]
\begin{center}
\footnotesize
\begin{tabular}
{cccllc}\cline{3-5} 5& \multicolumn{1}{|c}{$\strut\mkern-18mu$}&\multicolumn{1}{|l}{$\strut\:x_0\:$} & {$P(x_0)$} &
\multicolumn{1}{l|}{assumption$\mkern-2mu$} & \multicolumn{1}{l|}{$\strut\mkern-18mu$}\\ 6&
\multicolumn{1}{|c}{$\strut\mkern-18mu$}&\multicolumn{1}{|l}{} & {$P(x_0)\rightarrow Q(x_0)$} &
\multicolumn{1}{l|}{$\forall_x$e 2} & \multicolumn{1}{l|}{$\strut\mkern-18mu$}\\

9& \multicolumn{1}{|c}{$\strut\mkern-18mu$}&\multicolumn{1}{|l}{{\em then}$\mkern-4mu$} & {$Q(x_0)$} &
\multicolumn{1}{l|}{$\rightarrow$e $\ast$,5} & \multicolumn{1}{l|}{$\strut\mkern-18mu$}\\

7& \multicolumn{1}{|c}{$\strut\mkern-18mu$}&\multicolumn{1}{|l}{} & {$\forall_y\left(Q(x_0)\rightarrow R(y)\right)\mkern-8mu$}
& \multicolumn{1}{l|}{$\forall_x$e 3} & \multicolumn{1}{l|}{$\strut\mkern-18mu$}\\

8& \multicolumn{1}{|c}{$\strut\mkern-18mu$}&\multicolumn{1}{|l}{{\em then}$\mkern-4mu$} & {$Q(x_0)\rightarrow R(y_0)$} &
\multicolumn{1}{l|}{$\forall_y$e $\ast$} & \multicolumn{1}{l|}{$\strut\mkern-18mu$}\\

10& \multicolumn{1}{|c}{$\strut\mkern-18mu$}&\multicolumn{1}{|l}{{\em then}$\mkern-4mu$} & {$R(y_0)$} &
\multicolumn{1}{l|}{$\rightarrow$e $\ast$,9} & \multicolumn{1}{l|}{$\strut\mkern-18mu$}\\
\cline{3-5}
\end{tabular}
\begin{tabular}
{cccllc}\cline{3-5} 5& \multicolumn{1}{|c}{$\strut\mkern-18mu$}&\multicolumn{1}{|l}{$\strut\:x_0\:$} & {$P(x_0)$} &
\multicolumn{1}{l|}{assumption$\mkern-2mu$} & \multicolumn{1}{l|}{$\strut\mkern-18mu$}\\ 7&
\multicolumn{1}{|c}{$\strut\mkern-18mu$}&\multicolumn{1}{|l}{} & {$\forall_y\left(Q(x_0)\rightarrow R(y)\right)\mkern-8mu$} &
\multicolumn{1}{l|}{$\forall_x$e 3} & \multicolumn{1}{l|}{$\strut\mkern-18mu$}\\ 8&
\multicolumn{1}{|c}{$\strut\mkern-18mu$}&\multicolumn{1}{|l}{{\em then}$\mkern-4mu$} & {$Q(x_0)\rightarrow R(y_0)$} &
\multicolumn{1}{l|}{$\forall_y$e $\ast$} & \multicolumn{1}{l|}{$\strut\mkern-18mu$}\\ 6&
\multicolumn{1}{|c}{$\strut\mkern-18mu$}&\multicolumn{1}{|l}{} & {$P(x_0)\rightarrow Q(x_0)$} &
\multicolumn{1}{l|}{$\forall_x$e 2} & \multicolumn{1}{l|}{$\strut\mkern-18mu$}\\ 9&
\multicolumn{1}{|c}{$\strut\mkern-18mu$}&\multicolumn{1}{|l}{{\em then}$\mkern-4mu$} & {$Q(x_0)$} &
\multicolumn{1}{l|}{$\rightarrow$e $\ast$,5} & \multicolumn{1}{l|}{$\strut\mkern-18mu$}\\ 10&
\multicolumn{1}{|c}{$\strut\mkern-18mu$}&\multicolumn{1}{|l}{{\em then}$\mkern-4mu$} & {$R(y_0)$} &
\multicolumn{1}{l|}{$\rightarrow$e 8,$\ast$} & \multicolumn{1}{l|}{$\strut\mkern-18mu$}\\
\cline{3-5}
\end{tabular}
\end{center}
\caption{Two possible topological sortings of the proof
that is contained in the scope of the $x_0$ box in Fig.~\ref{some_proof}. These sortings have the largest number of ``then''
steps.}
\label{reorganisation}
\end{figure}

\par {It is desirable} to have a ``push-button'' tool that automatically finds a suitable topological sortings $\tau$ of a
natural deduction proof so that {they have} the largest
number of ``then'' $\tau$-steps but we show {here that this goal is non-trivial since the problem is NP-complete.}

\par
Naturally every sequence of ``then'' $\tau$-steps determines a linear fragment of reasoning in the proof script that we call
{a} {\em $\tau$-reasoning path}. More precisely, let us fix a directed path $P=\langle u_0,u_1, u_2, \dots, u_n \rangle$. Then
$P$ is a $\tau$-reasoning path if and only if $u_k$ is a ``then'' $\tau$-step for $k=1,2,\ldots,n$ ($u_0$ does not have to be a
``then'' $\tau$-step). To say that $P$ is maximal, means simply that $P$ is not a subsequence of any other $\tau$-reasoning
path. It is easy to check that the number of ``then'' $\tau$-steps depends on the number of all maximal $\tau$-reasoning paths.
To be more precise, the sum of the number of ``then'' $\tau$-steps and the number of maximal $\tau$-reasoning paths is constant
and equal to the number of reasoning steps. Hence the number of ``then'' $\tau$-steps is maximal if and only if the number of
all maximal $\tau$-reasoning paths is minimal.

This means we can equivalently formulate many definitions in terms of the number of ``then'' $\tau$-steps and in terms of
$\tau$-reasoning paths. In addition, we would like the flow of information between any two different maximal $\tau$-reasoning
paths {to be minimal}. {This is equivalent to the property that the number of {\em internal $\tau$-arcs}, i.e.\ the reference
arcs whose head and tail belong to the same maximal $\tau$-reasoning path, is maximal. It is true since the sum of all
reference arcs between any two different maximal $\tau$-reasoning paths {plus} the number of internal $\tau$-arcs is constant
and equal to the number of reference arcs. }

As an illustration of ``then'' $\tau$-steps and $\tau$-reasoning paths, let us consider a short quasi proof script written in
the Mizar style. A theorem in the Mizar style has the form:
{\small
\begin{equation}
\begin{tabular}{l}
\miz{theorem} {\em label}\miz{:} {\em statement}\\
\miz{proof}\\
$\strut\mkern20mu${\em proof steps}\\
\miz{end};\\
\end{tabular}
\end{equation}
}
and a proof step in general has the form:
{\small
\begin{equation}\label{step_private}
\begin{tabular}{c}
{\em label}\miz{:} {\em statement} \miz{by} {\em justification} \miz{;}\\
\end{tabular}
\end{equation}
}
Let us focus on a statement of the theory of fields that the inverse of a product of two elements of the field is equal to the
product of their inverses. Note that symbols \mbox{1.F}, \mbox{0.F} represent the additive identity element and the
multiplicative identity element of {the} field \mbox{F}, respectively. Note additionally that the system Mizar uses only ASCII
characters. Operations such as ``{\C}{}'', ``{\M}{}'' are represented in Mizar as \miz{"} and \miz{*}, respectively, and
VECTSP\!\_1,
\begin{figure}[h!]
\small
\begin{tabular}{rl}
&\miz{theorem} Th1\miz{:}\\ &\miz{for} F \miz{be} Field\miz{,} a\miz{,} b \miz{be} Element \miz{of} F \miz{st} a \miz{<>} 0.F
\miz{\&} b \miz{<>} 0.F \\ &$\strut\mkern20mu$\miz{holds} \miz{(}a {{\M}} b\miz{)}{\C} \miz{=} a{\C} {{\M}} b{\C}\\
&\miz{proof}\\
\QQ{1}&$\strut\mkern20mu$\miz{let} F \miz{be} Field\miz{,} a\miz{,} b \miz{be} Element \miz{of} F\miz{;}\\
\QQ{2}&$\strut\mkern20mu$\miz{assume} A1\miz{:} a \miz{<>} 0.F \miz{\&} b \miz{<>} 0.F\miz{;}\\
\QQ{3}&$\strut\mkern20mu$A2\miz{:} a{\C} + \miz{(}a{\C} {{\M}} 0.F\miz{)} \miz{=} \miz{(}a{\C} {{\M}} 1.F\miz{)} + \miz{(}a{\C}
{{\M}} 0.F\miz{)} \miz{by} VECTSP\!\_1\miz{:def} 6\miz{;}\\
\QQ{4}&$\strut\mkern20mu$A3\miz{:} \miz{(}a {{\M}} b\miz{)} {{\M}} \miz{(}b {\C}{{\M}} a{\C}\miz{)} \miz{=} \miz{((}a {{\M}}
b\miz{)} {{\M}} b{\C}\miz{)} {{\M}} a{\C} \miz{by} GROUP\!\_1\miz{:def} 3\miz{;}\\
\QQ{5}&$\strut\mkern20mu$A4\miz{:} a{\C} + \miz{(}a{\C} {{\M}} 0.F\miz{)} \miz{=} a{\C} {{\M}} \miz{(}1.F + 0.F\miz{)} \miz{by}
A2\miz{,}VECTSP\!\_1\miz{:def} 2\miz{;}\\
\QQ{6}&$\strut\mkern20mu$A5\miz{:} b \miz{=} 1.F {{\M}} b \miz{by} VECTSP\!\_1\miz{:def} 6\miz{;}\\
\QQ{7}&$\strut\mkern20mu$A6\miz{:} a{\C} + \miz{(}a{\C} {{\M}} 0.F) \miz{=} a{\C} {{\M}} 1.F \miz{by} A4\miz{,}
RLVECT\!\_1\miz{:def} 4\miz{;}\\
\QQ{8}&$\strut\mkern20mu$A7\miz{:} \miz{(}a {\M
} b\miz{)} {{\M}} \miz{(}b{\C} {{\M}} a{\C}\miz{)} = \miz{(}a {{\M}} \miz{(}b {{\M}} b{\C}\miz{)} {{\M}} a{\C}\miz{)} \miz{by}
A3\miz{,} GROUP\!\_1\miz{:def} 3\miz{;}\\
\QQ{9}&$\strut\mkern20mu$A8\miz{:} a{\C} + \miz{(}a{\C} {{\M}} 0.F\miz{)} \miz{=} a{\C} \miz{by} A6\miz{,}
VECTSP\!\_1\miz{:def} 6\miz{;}\\
\QQ{10}&$\strut\mkern20mu$A9\miz{:} b \miz{=} \miz{(}a{\C} {{\M}} a\miz{)} {{\M}} b \miz{by} A1\miz{,} A5\miz{,}
VECTSP\!\_1\miz{:def} 10\miz{;}\\
\QQ{11}&$\strut\mkern20mu$A10\miz{:} \miz{(}-(a{\C})+a{\C}\miz{)} + \miz{(}a{\C} {{\M}} 0.F\miz{)} \miz{=} -\miz{(}a{\C}\miz{)}
+ \miz{(}a{\C}\miz{)} \miz{\&} -\miz{(}a{\C}\miz{)} + a{\C} \miz{=} 0.F\\
&$\strut\mkern40mu$\miz{by} A8\miz{,} RLVECT\!\_1\miz{:def} 3\miz{,} \miz{def} 10\miz{;}\\
\QQ{12}&$\strut\mkern20mu$A11\miz{:} \miz{(}a {{\M}} b\miz{)} {\M
} \miz{(}b{\C} {{\M}} a{\C}\miz{)} \miz{=} a {{\M}} 1.F {{\M}} a{\C}  \miz{by} A7\miz{,} A1\miz{,} VECTSP\!\_1\miz{:def}
10\miz{;}\\
\QQ{13}&$\strut\mkern20mu$A12\miz{:} a{\C} {{\M}} 0.F \miz{=} 0.F \miz{by} A10\miz{,} RLVECT\!\_1\miz{:def} 4\miz{;}\\
\QQ{14}&$\strut\mkern20mu$A13\miz{:} \miz{(}a {{\M}} b\miz{)} {\M
} \miz{(} b{\C} {{\M}} a{\C}\miz{)} \miz{=} a {{\M}} a{\C} \miz{by} A11\miz{,} VECTSP\!\_1\miz{:def} 6\miz{;}\\
\QQ{15}&$\strut\mkern20mu$A14\miz{:} a{\C} {{\M}} 0.F \miz{<>} a{\C} {{\M}} \miz{(}a {{\M}} b\miz{)} \miz{by} A12\miz{,}
A1\miz{,} A9\miz{,} GROUP\!\_1\miz{:def} 3\miz{;}\\
\QQ{16}&$\strut\mkern20mu$\miz{(}a {{\M}} b\miz{)} {{\M}} \miz{(}b{\C} {{\M}} a{\C}\miz{)} \miz{=} 1.F \miz{by} A13\miz{,}
A1\miz{,} VECTSP\!\_1\miz{:def} 10\miz{;}\\
\QQ{17}&$\strut\mkern20mu$\miz{thus} \miz{then} \miz{thesis} \miz{by} A14\miz{,} VECTSP\!\_1\miz{:def} 10{;}\\
&\miz{end;}\\
\end{tabular}
\caption{An example of illegible proof script written in the Mizar style.}
\label{Mizar_proof}
\end{figure}
GROUP\!\_1, RLVECT\!\_1 are identifiers of Mizar articles
\cite{VECTSP_1,GROUP_1,RLVECT_1}. The proof script presented in
Fig.~\ref{Mizar_proof} does not comply well to the Behaghel's First Law. Only one step (line 17) refers to the preceding step,
the other steps in the reasoning refer to steps placed far away in the proof script. Analyzing the possible topological
sortings of the corresponding proof graph (see Fig.~\ref{Mizar_proof_graph}), 359598 cases, we get that the greatest number of
``then'' $\tau$-steps is reached in 16 linearizations and is equal to 12. Hence the smallest number of maximal $\tau$-reasoning
paths is equal to 5. In 6 cases of these 16, we have also the smallest number of reference arcs between any two different
maximal $\tau$-reasoning paths, which is equal to 5. One of these 6 cases is presented in Fig.~\ref{Mizar_mod}.
\begin{figure}[ht!]
\begin{center}
\begin{tikzpicture}
[>=stealth,->,shorten >=2pt,looseness=.5,auto]
    \matrix (M)[matrix of math nodes,
        column sep={1.2cm,between origins},
        row sep={1.2cm,between origins},
        nodes=myrectanglenodes]
        {
&|(A7)|{\:7:\:}&|(A5)|{\:5:\:}&|(A3)|{\:3:\:}&\\ |(A9)|{\:9:\:}&|(A6)|{\:6:\:}&|(A1)|{\:1:\:}&|(A8)|{\:8:\:}&|(A4)|{\:4:\:}\\
|(A11)|{\:11:\:}&|(A10)|{10:}&|(A2)|{\:2:\:}&|(A12)|{12:}&|(A14)|{\:14:\:}\\
|(A13)|{13:}&&|(A15)|{15:}&|(A17)|{17:}&|(A16)|{16:}\\
        };
\draw[Referencja](A3)--(A5);
\draw[Referencja](A5)--(A7);
\draw[Referencja](A7)--(A9);
\draw[Referencja](A9)--(A11);
\draw[Referencja](A11)--(A13);
\draw[Referencja](A13)--(A15);
\draw[Referencja](A15)--(A17);
\draw[Referencja](A16)--(A17);
\draw[Referencja](A14)--(A16);
\draw[Referencja](A12)--(A14);
\draw[Referencja](A8)--(A12);
\draw[Referencja](A4)--(A8);
\draw[Referencja](A10)--(A15);
\draw[Referencja](A6)--(A10);
\draw[Referencja](A2)--(A15);
\draw[Referencja](A2)--(A12);
\draw[Referencja](A2)--(A10);
\draw[Referencja](A2)--(A16);
\draw[Variable](A1)--(A3);
\draw[Variable](A1)--(A5);
\draw[Variable](A1)--(A7);
\draw[Variable](A1)--(A6);
\draw[Variable](A1)--(A10);
\draw[Variable](A1)--(A2);
\draw[Variable](A1)--(A12);
\draw[Variable](A1)--(A8);
\draw[Variable,bend right, looseness=0.8]($(A1.south west)+(0,0.14)$)to($(A11.east)+(0,0.07)$);
\draw[Variable,bend left, looseness=0.8]($(A1.south east)+(0,0.14)$)to($(A14.west)+(0,0.07)$);
\draw[Variable,bend right, looseness=1](A1)to(A9);
\draw[Variable,bend left, looseness=1](A1)to(A4);
\draw[Variable,bend right, looseness=1](A1)to(A15);
\draw[Variable,bend right, looseness=0.8]($(A1.south west)+(0,0.08)$)to(A13);
\draw[Variable,bend left, looseness=0.8]($(A1.south east)+(0,0.08)$)to(A16);
\end{tikzpicture}
\end{center}
\caption{The abstract proof graph illustrating
the structure of the reasoning presented in Fig.~\ref{Mizar_proof}.}
\label{Mizar_proof_graph}
\end{figure}

This example illustrates the significance of ``then'' construction in the process of improving proof scripts legibility. This
construction in natural way focuses the reader's attention on sequences of ``then'' $\tau$-steps. Note that every sequence of
steps that has a label connected only to the last step of this sequence (as is the case of the sequence of ``then''
$\tau$-steps starting in step 3 and ending in step 13 in Fig.~\ref{Mizar_mod}) contains additional information for readers.
Naturally, the previous unlabelled steps have only local significance, i.e.\ each of them is used exclusively  to justify the
directly following step. This is a sign for readers that they can ``encapsulate'' the piece of reasoning and abstract away from
its content in many situations.

\begin{figure}[h!]
\small
\begin{tabular}{rl}
&\miz{theorem} Th1\miz{:}\\ &\miz{for} F \miz{be} Field\miz{,} a\verb*V,V b \miz{be} Element \miz{of} F \miz{st} a \miz{<>} 0.F
\miz{\&} b \miz{<>} 0.F \\ &$\strut\mkern20mu$\miz{holds} \miz{(}a {{\M}} b\miz{)}{\C} \miz{=} a{\C} {{\M}} b{\C}\\
&\miz{proof}\\
\QQ{1}&$\strut\mkern20mu$\miz{let} F \miz{be} Field\miz{,} a\miz{,} b \miz{be} Element \miz{of} F\miz{;}\\
\QQ{3}&$\strut\mkern20mu$a{\C} + \miz{(}a{\C} {\M} 0.F\miz{)} \miz{=} \miz{(}a{\C} {\M} 1.F\miz{)} + \miz{(}a{\C} {\M}
0.F\miz{)} \miz{by} VECTSP\!\_1\miz{:def} 6\miz{;}\\
\QQ{5}&$\strut\mkern20mu$\miz{then} a{\C} + \miz{(}a{\C} {\M} 0.F\miz{)} \miz{=} a{\C} {\M} \miz{(}1.F + 0.F\miz{)} \miz{by}
VECTSP\!\_1\miz{:def} 2\miz{;}\\
\QQ{7}&$\strut\mkern20mu$\miz{then} a{\C} + \miz{(}a{\C} {\M} 0.F\miz{)} \miz{=} a{\C} {\M} 1.F \miz{by} RLVECT\!\_1\miz{:def}
4\miz{;}\\
\QQ{9}&$\strut\mkern20mu$\miz{then} a\C
 + \miz{(}a{\C} {\M} 0.F\miz{)} \miz{=} a{\C} \miz{by} VECTSP\!\_1\miz{:def} 6\miz{;}\\
\QQ{11}&$\strut\mkern20mu$\miz{then} \miz{(}-\miz{(}a{\C}\miz{)} + a{\C}\miz{)} + \miz{(}a{\C} {\M} 0.F\miz{)} \miz{=}
-\miz{(}a{\C}\miz{)} + \miz{(}a{\C}\miz{)} \miz{\&} -\miz{(}a{\C}\miz{)} + a{\C} \miz{=} 0.F\\
&$\strut\mkern40mu$\miz{by} RLVECT\!\_1\miz{:def} 3\miz{,} \miz{def} 10\miz{;}\\
\QQ{13}&$\strut\mkern20mu$\miz{then} A1\miz{:} a{\C} {\M} 0.F \miz{=} 0.F \miz{by}  RLVECT\!\_1\miz{:def} 4\miz{;}\\
\QQ{6}&$\strut\mkern20mu$A2\miz{:} b \miz{=} 1.F {\M} b \miz{by} VECTSP\!\_1\miz{:def} 6\miz{;}\\
\QQ{2}&$\strut\mkern20mu$\miz{assume} A3\miz{:} a \miz{<>} 0.F \miz{\&} b \miz{<>} 0.F\miz{;}\\
\QQ{10}&$\strut\mkern20mu$\miz{then} b \miz{=} \miz{(}a{\C} {\M} a\miz{)} {\M} b \miz{by} A2\miz{,} VECTSP\!\_1\miz{:def}
10\miz{;}\\
\QQ{15}&$\strut\mkern20mu$\miz{then} A4\miz{:}a{\C} {\M} 0.F \miz{<>} a{\C} {\M} \miz{(}a {\M} b\miz{)} \miz{by} A1\miz{,}
A3\miz{,} GROUP\!\_1\miz{:def} 3\miz{;}\\
\QQ{4}&$\strut\mkern20mu$\miz{(}a {\M} b\miz{)} {\M} \miz{(}b{\C} {\M} a{\C}\miz{)} \miz{=} \miz{((}a {\M} b\miz{)} {\M}
b{\C}\miz{)} {\M} a{\C} \miz{by} GROUP\!\_1\miz{:def} 3\miz{;}\\
\QQ{8}&$\strut\mkern20mu$\miz{then} \miz{(}a {\M} b\miz{)} {\M} \miz{(}b{\C} {\M} a{\C}\miz{)} \miz{=} \miz{(}a {\M} \miz{(}b
{\M} b{\C}\miz{)} {\M} a{\C}\miz{)} \miz{by} GROUP\!\_1\miz{:def} 3\miz{;}\\
\QQ{12}&$\strut\mkern20mu$\miz{then} \miz{(}a {\M} b\miz{)} {\M} \miz{(}b{\C} {\M} a{\C}\miz{)} \miz{=} a {\M} 1.F {\M} a{\C}
\miz{by} A3\miz{,} VECTSP\!\_1\miz{:def} 10\miz{;}\\
\QQ{14}&$\strut\mkern20mu$\miz{then} \miz{(}a {\M} b\miz{)} {\M} \miz{(}b{\C} {\M} a{\C}\miz{)} \miz{=} a {\M} a{\C} \miz{by}
VECTSP\!\_1\miz{:def} 6\miz{;}\\
\QQ{16}&$\strut\mkern20mu$\miz{then} \miz{(}a {\M}{} b\miz{)} {\M} \miz{(}b{\C} {\M} a{\C}\miz{)} \miz{=} 1.F  \miz{by}
A3\miz{,} VECTSP\!\_1\miz{:def} 10\miz{;}\\
\QQ{17}&$\strut\mkern20mu$\miz{thus} \miz{then} \miz{thesis} \miz{by} A4\miz{,} VECTSP\!\_1\miz{:def} 10\miz{;}\\
&\miz{end;}\\
\end{tabular}
\caption{A linearisation of the proof presented in Fig.~\ref{Mizar_proof}. This linearisation maximizes the number of ``then''
$\tau$-steps.}
\label{Mizar_mod}
\end{figure}

The following short analysis of this property shows that it is impossible to apply several proof readability criteria at the
same time and that they are sometimes contradictory. A~commonly accepted solution in such circumstances is to impose a
hierarchy on the optimized properties (this additional property and several others are briefly described in
Section~\ref{Other}).

Until now, we have concentrated on the number of steps that refer to the preceding step. Maximization of sequences of this kind
is not the only strategy to improve readability. The topological sorting presented in Fig.~\ref{Mizar_mod} has additional
important features that are not consequence of Behaghel's First Law, but are related to readability. In  particular, it is one
of two linearizations
 that have the smallest sum of {all} distances between premises and their uses, equal to 38 (the second topological sorting
 with this property contains 11 ``then'' $\tau$-steps).
Additionally, 4 labels have been used in this reasoning, but it is not the smallest number of labels in existing linearisation.
If we swap lines 6 and 2, then the label \mbox{A2} is not necessary. Among all sortings of the abstract proof graph in
Fig.~\ref{Mizar_proof_graph} the smallest possible number of labels is 3. This is achieved in 8 cases. However, none of the
cases minimizes the sum of all distances since the sum there is at best 41. Notably there are, among the 8 cases, ones that
contain exactly 12 ``then'' $\tau$-steps, which is the maximal possible number for the abstract proof graph. This illustrates
that there are proofs for which all plausible readability optimalization criteria cannot be met at once.

Another way to impose a restriction on the distance between premises and their uses is
 to minimize the maximal distance between an arbitrary premise and its use. This distance for the reasoning presented in
 Fig.~\ref{Mizar_mod} is equal to 7 (the distance is reached between lines 2 and 16 for the label \mbox{A3}).
 The smallest restriction of maximal distance in this abstract proof graph is equal to 5 and is obtained in 128 cases, but
 optimal values of
 previously mentioned parameters
  are not always possible to reach among them. More precisely,
  the number of ``then'' $\tau$-steps is at most 12 (the optimal number is 12),
  the number of reference arcs
between any two different maximal $\tau$-reasoning paths is at least 7 (the optimal number is 5), the sum of distances between
premises and their uses is at least 39 (the optimal number is 38), and the number of labels in a reasoning is at least 6 (the
optimal number is 3).

\section{Formulation of Behaghel's Law
Determinants}\label{Formulation_of_the_Behaghel's_Law_Problem} {To formulate the criterion we need to set the appropriate
vocabulary of notations.} Let $G=\langle V,E\rangle$ be a \DAG\ and let $E_1$ be a subset of $E$. For a subset $V_1$ of $V$ we
denote by $\DAGcut{G}{V_1}$ the subgraph of $G$ induced by $V_1$. For a vertex $u$ in $V$, we use the following notation:
\begin{equation}
\begin{array}{lcrcr}
N^{-}_{G}(u) & \coloneq & \{v\in V:\langle v,u\rangle\in E\}&&(
\mbox{incoming arcs}),\\
N^{+}_{G}(u) & \coloneq & \{v\in V:\langle u,v\rangle\in E\} &\strut\mkern40mu&(\mbox{outgoing arcs}),
\end{array}
\end{equation}
$|N^{-}_{G}(u)|$ is the {\em in-degree} of $u$ and $|N^{+}_{G}(u)|$ is the {\em out-degree} of $u$. An arc is called {\em
$E_1$-arc} if it belongs to $E_1$. A path $P= \langle u_1,u_2,\ldots, u_{n}\rangle $ of $G$ is called an {\em $E_1$-path} if
$\langle u_i,u_{i+1}\rangle$ is an $E_1$-arc for $i=1,2,\ldots,n-1$. The concatenation of two paths $P_1,$ $P_2$ is denoted by
$P_1\!\mathbin{{}^\frown}\!P_2$. The length of $P$ is denoted by $l(P)$ and the set $\{u_1,u_2,\ldots, u_{n}\}$ of vertices of
$P$ is denoted by $\vertexOF{P}$. An $E_1$-arc $\langle v,u\rangle$ is an {\em $E_1$-shortcut} if there exists an $E_1$-path
that leads from $v$ to $u$ and whose length is at last 2.

\par
Now we define a formal equivalent of a $\tau$--reasoning path, called a $\tauE{\tau}{E_1}$-path. Given $\tau\in \TS(G)$, an
${E_1}$-path $P= \langle u_1,u_2,\ldots, u_{n}\rangle $ is called a {\em $\tauE{\tau}{E_1}$-path} if
$\tau(u_i)+1=\tau(u_{i+1})$ for $i=1,2,\ldots,n-1$. Additionally, $P$ is said to be a maximal {$\tauE{\tau}{E_1}$-path} if and
only if $\vertexOF{P}$ is not included in any other set of vertices of {a} $\tauE{\tau}{E_1}$-path. Let
\begin{equation}\pipart{\pi}{\tau}{E_1}\coloneq\{\vertexOF{P}\mid P \mbox{ is a maximal } \tauE{\tau}{E_1}\mbox{-path} \}
\end{equation}
and call it {\em {a} partition of $G$ determined by $\tau$ with {respect to $E_1$}}. For subsets $V_1,V_2$ of $V$ we use the
following notation:
\begin{equation}
V_1\pppp{E_1}V_2\coloneq\{\langle v_1,v_2\rangle\in E_1\mid v_1\in V_1\,\wedge\,v_2\in V_2\}.
\end{equation}

\par The methods of improving legibility of proofs described above
can be formulated as the following two {decision} problems:\\

\nopar{\bf 1st Method of Improving Legibility} ({\bf$1$st MIL}):\\
$\strut\mkern25mu$\textsc{Instance}: A \DAG\: $G = \langle V,E\rangle$, a subset $E_1$ of $E$, and a positive integer $K\leq
|V|$.\\ $\strut\mkern25mu$\textsc{Question}: Does there exist a topological sorting $\tau$ of $G$ for which $\pi^{E_1}_\tau$
has size\break at most $K$?\\

\nopar{\bf 2nd Method of Improving Legibility} ({\bf $2$nd MIL}):\\
$\strut\mkern25mu$\textsc{Instance}: A \DAG\: $G = \langle V,E\rangle$, a subset $E_1$ of $E$, and a positive integer $L\leq
|E_1|$. \\ $\strut\mkern25mu$\textsc{Question}: Does there exist a topological sorting $\tau$ of $G$ for which
$$\strut\mkern-5mu
\underset{\scriptsize
\begin{array}{c}P_1,P_2\in\pi^{E_1}_\tau\\
P_1\neq P_2\\
\end{array}}{\sum}
|P_1\pppp{E_1}P_2|\leq L\mbox{?}$$\\

\noindent
In our setting, the subset $E_1$ corresponds to the set of reference arcs. The \mbox{$1$st MIL} corresponds to searching for
 a topological sorting $\tau$ that has
 the smallest possible number of
 maximal sequences of ``then'' $\tau$-steps.
Recall that we can formulate this problem in the equivalent way as maximization of the number of ``then'' $\tau$-steps.

The \mbox{$2$nd MIL} is an extension of the idea of the first method. We demand that every ``then'' $\tau$-step not only refers
to the preceding step in a maximal sequence of ``then'' $\tau$-steps, but also that other references of this step point mostly
to steps of this sequence. To fulfil this condition we search for a topological sorting $\tau$ that has at most a given number
of references between any two different maximal sequences of ``then'' $\tau$-steps.

\par Now we show that these problems have the same answer for {a} special type of \DAG{}s and
parameters $K,\,L$ such that $K=L+|V|-|E_1|$.

\begin{lem}\label{KP_PROP1}
Let $G=\langle V,E\rangle$ be a \DAG\: and $E_1$ be a subset of $E$ for which $G$ is without $E_1$-shortcuts. For any $\tau\in
\TS(G)$, the following holds:
\begin{equation}
|V|-|\pipart{\pi}{\tau}{E_1}|=|E_1|-\underset{
\scriptsize\begin{array}{c} P_1,P_2\in\pipart{\pi}{\tau}{E_1}\\
P_1\neq P_2\\
\end{array}}{\sum}\mkern-10mu |P_1\pppp{E_1}P_2|.
\end{equation}
\end{lem}
\begin{proof}
Since $G$ is without $E_1$-shortcuts, we have $|P|-1=|P\pppp{E_1}P|$ for $P\in \pi^{E_1}_\tau$. Consequently,
$|E_1|=\underset{\scriptsize\begin{array}{c} P_1,P_2\in\pipart{\pi}{\tau}{E_1}\\ P_1\neq P_2\\\end{array}}{\sum}
|P_1\pppp{E_1}P_2| +\underset{P\in\pipart{\pi}{\tau}{E_1}}{\sum}(|P|-1),$ but
$\underset{P\in\pipart{\pi}{\tau}{E_1}}{\sum}(|P|-1) = |V|-|\pipart{\pi}{\tau}{E_1}|$.
\end{proof}

This lemma makes it possible to give a single construction that shows hardness of both problems. We exploit it later in the
next section.

\section{Acyclic Hamiltonian Partition}\label{Acyclic_Hamiltonian_Partition}

The {notion} of topological sorting is difficult to manipulate and makes various formulations obscure. Therefore it is useful
to express these problems in terms of {\em acyclic partitions}. Following Borowiecki and Mih\'ok~\cite{BM} we use the following
notation. Let $\mathcal{P}$ be a graph property. A {\it$\mathcal{P}^k$-partition of a graph} $G=\langle V,E\rangle$ is a
partition $\pi=\{V_1,V_2,\ldots,V_k\}$ of $V$ such that $\DAGcut{G}{V_i}$ has property $\mathcal{P}$ for $i=1,2,\ldots,k$. A
partition $\phi$ of $G$ is called a {\em $\mathcal{P}^\ast$-partition} if there exists a positive integer $k$ such that $\phi$
is a $\mathcal{P}^k$-partition. Let us denote by $\pidigraph{G}{\pi}$ the directed graph $\langle \pi,\{\langle V_i,V_j\rangle:
1\leq i,j\leq k\wedge i\neq j\wedge  V_i\pppp{E}V_j\neq \emptyset\}\rangle$. We say that $\mathcal{P}^k$-partition $\pi$ is
acyclic if and only if $\pidigraph{G}{\pi}$ is acyclic. For simplicity, an acyclic $\mathcal{P}^k$-partition is called a {\em
$\mathcal{P}^{(k)}$-partition} and an acyclic $\mathcal{P}^\ast$-partition is called a {\em $\mathcal{P}^{(\ast)}$-partition}.
Let $G=\langle V,E\rangle$ be a \DAG{} and let $V_1$ be a subset of $V$. We define the property $\mathcal{H}$ of $G$ as: the
subgraph $\DAGcut{G}{V_1}$ has a Hamiltonian path.

\par Our main object of investigation is now the following problem:\\

\nopar{\bf Acyclic Hamiltonian Partition} ({\bf\AHP}):\\
$\strut\mkern25mu$\textsc{Instance}: A \DAG\: $G = \langle V,E\rangle$ and a positive integer $K\leq |V|$.\\
$\strut\mkern25mu$\textsc{Question}: Is there an $\mathcal{H}^{(\ast)}$-partition of $G$ of size at most $K$?\\

\noindent It is evident that  $\pipart{\pi}{\tau}{E}$ is an $\mathcal{H}^{(|\pipart{\pi}{\tau}{E}|)}$-partition of $G$
for any $\tau\in \TS(G)$. We show that an $\mathcal{H}^{(\ast)}$-partition $\pi$ of $G$ which solves \AHP{} determines the
topological sorting $\tau$ of $G$ such that $|\pipart{\pi}{\tau}{E}|\leq |\pi|$. In consequence we can consider the $1$st
\mbox{MIL} problem to be equivalent to \AHP{} in the case $E=E_1$. Hence if we show that \AHP{} is an NP-complete problem for \DAG{} without
$E$-shortcuts, then $1$st \mbox{MIL} is also NP-complete for \DAG{} without $E_1$-shortcuts, where $E_1=E$. Additionally for so
restricted instances, $1$st \mbox{MIL} and $2$nd \mbox{MIL} are equivalent by Lemma~\ref{KP_PROP1}. Based on this observation,
to show that these methods are NP-complete, we focus in Section~\ref{NPproof} only on the complexity of \AHP{}.

\par To shorten our notation, $\pipath{\pi}{p}$ denotes the Hamiltonian path of $\DAGcut{G}{p}$, where $p\in \pi$ and $\pi$ is
an $\mathcal{H}^{\ast}$-partition.
Additionally, we denote by $\pipath{\pi}{v}$ the path $\pipath{\pi}{p}$, where $v \in p$.

\begin{lem}\label{KP_LEM1}
Let $\pi$ be an $\mathcal{H}^{(\ast)}$-partition and suppose that $\pi$ solves \AHP. Then there exists $\varrho\in \TS(G)$ such
that $|\pipart{\pi}{\varrho}{E}|\leq |\pi|$.
\end{lem}
\begin{proof}
{As $\pi$ is acyclic we can take} $\tau\in \TS(\pidigraph{G}{\pi})$. {Let} $\pipath{\pi}{P}=(u_1^P,u_2^P,\ldots,u_{|P|}^P)$,
for $P\!\in\! \pi$. {Let $\varrho:\vertexOF{G}\rightarrow\{1,2,\ldots,|\vertexOF{G}|\}$ be a function} given by the formula
$\varrho(u_i^Q) = i\!+\!\underset{R\in\, \pi:\tau(R)<\tau(Q)}{\sum}|R|$, where $1\leq i\leq |Q|$, $Q\in \pi$. {Observe that
$\varrho$ is a topological sorting of $G=\langle V,E\rangle$}. Obviously, every Hamiltonian path $\pipath{\pi}{P}$ is also a
$\tauE{\varrho}{E}$-path, for $P \in \pi$. Therefore, the number of all maximal $\tauE{\varrho}{E}$-paths is not greater than
the cardinality of $\pi$.
\end{proof}

\section{The NP-completeness of the \AHP{} problem} \label{NPproof}
It is clear that \AHP{} is in NP. We can guess a partition and verify that it is indeed an $\mathcal{H}^{(\ast)}$-partition. We
give here a decision problem for directed graphs and its transformation to \AHP{} that proves that \AHP{} is {NP-hard}. Given
$G=\langle V,E \rangle$, a {\em feedback arc set} is a subset $E^\prime$ of $E$ {that} contains at least one arc from each
directed cycle of $G$. The following problems are NP-complete (see \mbox{GT1}, \mbox{GT8} in \cite{NPLIST}).\\

\nopar{\bf Feedback Arc Set} ({\FAS}):\\
$\strut\mkern25mu$\textsc{Instance}: A directed graph $G = \langle V,E\rangle$ and a positive integer $K\leq |V|$.\\
$\strut\mkern25mu$\textsc{Question}: Is there a feedback arc set of $G$ of size at most $K$?\\

\nopar{\bf Vertex Cover} (\VS):\\
$\strut\mkern25mu$\textsc{Instance}: An undirected graph $G = \langle V,E\rangle$ and a positive integer $K\leq |V|$.\\
$\strut\mkern25mu$\textsc{Question}: Is there a vertex cover of size at most $K$, i.e.\ a subset $V^\prime\subseteq V$ with
$|V^\prime|\leq K$ such that for each edge $\{u,v\}\in E$ at least one of $u$ or $v$ belongs to $V^\prime$?\\

\par Transformation of \VS{} to \FAS{} has been considered by R. M. Karp \cite{RACP}.
Using this transformation we can prove that \FAS{} is NP-complete even for digraphs in which all vertices have in-degree or
out-degree equal to $1$.
\begin{lem}\label{KP_PROP2}
\FAS{} is  NP-complete even for directed graphs without self-loops, in which additionally all vertices have in-degree or
out-degree equal to $1$.
\end{lem}
\begin{proof}
Given an undirected graph $G=\langle V,E \rangle$, we follow R. M. Karp and construct the directed graph $G^\prime = \langle
V^\prime, E^\prime\rangle$ as
\begin{equation}
\begin{array}{rcl}
V^\prime &=& V\times \{0,1\},\\ E^\prime&=&\{\langle \langle v,0\rangle,\langle v,1\rangle\rangle:v\in V\}\cup
\{\langle \langle u,1\rangle,\langle v,0\rangle\rangle:\{u,v\}\in E\}.
\end{array}
\end{equation}
There are no self-loops as for $\langle\langle u,i\rangle,\langle v,j\rangle\rangle\in E^\prime$ we have $i\neq j$. The
out-degree of vertices $\langle u,0\rangle$ is 1 and in-degree of vertices $\langle v,1\rangle$ is 1. The directed graph
$G^\prime$ has the property formulated in the proposition.
\end{proof}

\par
We transform the subcase of \FAS{} that is specified in Lemma~\ref{KP_PROP2} to \AHP{}. Let $G=\langle V,E \rangle$ be a
directed graph without self-loops, in which additionally all vertices have in-degree or out-degree equal to $1$.
The case $V=\emptyset$ is obvious, therefore we can assume that $|V|>0$. Note that since there exists at last one vertex with
in-degree or out-degree equal to $1$, we have $|E|>0$.   Let us fix a one-to-one function $\e:E\longrightarrow \{1,2,\ldots,|E|\}$. We construct a \DAG{}
$\DD{G}{\e}=\langle V_G,E_G\rangle$ without $E_G$-shortcuts (Def.~\ref{DAG}) and we prove that there exists $\F$ that is a
feedback arc set of $E$ which has size at most $K$ if and only if there exists an $\mathcal{H}^{(\ast)}$-partition $\pi_\F$ of
$\DD{G}{\e}$ with size at most $|V|\cdot(|E|+1)+K$. The set $V_G$ of vertices is {divided} in two separate classes,  grouped
according to their intended function: ``gadgets'', {each of which} corresponds to a single vertex of $V$, and ``glues'', which
correspond to a single arc from $E$. We divide this construction in
 sequence of steps and statements, included into three subsections. In Subsection \ref{OG} we introduce the notion of gadgets
 and prove their  {essential}  properties.
Next in Subsection \ref{SS2} we complete the definition of $\DD{G}{\e}$ by presentation of ``glues'' and construct  $\pi_\F $
using $\F$. At the end of Subsection \ref{SS3} we construct {the feedback arc set} $\F$ using $\pi_\F$.

\par Let us fix the notation $\G=\langle \V,\E\rangle$ for a
directed graph without self-loops, in which all vertices have in-degree or out-degree equal to $1$.

\subsection{The Gadgets}\label{OG}

Let us take a vertex $r\in \V$. We call a graph $\gadget_r=\langle \gadgetV_r,\gadgetE_r\rangle$ ({see Fig.~\ref{G1}}) the
  {\em gadget corresponding to $r$}, {and define it} as:
\begin{equation}
\begin{array}{rcl}
\gadgetV_r &=& \{r_{i,j}:0\leq i,j\leq |\E|\},\\
\gadgetE_r &=& \LL_r\cup \RR_r,\\
\end{array}\label{Gadget_def_5.2}
\end{equation}
where
\begin{equation}
\begin{array}{lcl}
\LL_r&\coloneq&\{\langle r_{i,j},r_{i,j+1}\rangle:0\leq i,j< |\E|\}\cup
\{\langle r_{|\E|,i},r_{|\E|,i+1}\rangle: 0\leq i<|\E|\},\\
\RR_r&\coloneq&\{\langle r_{i,j},r_{i+1,j}\rangle:0\leq i,j< |\E|\}\cup
\{\langle r_{i,|\E|},r_{i+1,|\E|}\rangle: 0\leq i<|\E|\}.
\\
\end{array}
\end{equation}
Let us introduce notations for $\gadgetE_r$--paths:
\begin{equation}
\begin{array}{lcrclcr}
\LL_r^{i}&\coloneq&(r_{i,0},r_{i,1},\ldots,r_{i,|\E|}),&\strut\:\:&
\RR_r^{i}&\coloneq&(r_{0,i},r_{1,i},\ldots,r_{|\E|,i}),\\
\end{array}
\end{equation}
where $0\leq i\leq |\E|$, and for $\mathcal{H} ^{(|\E|+1)}$-partitions of $\gadget_r$:
\begin{equation}
\begin{array}{lcrclcr}
\mathcal{L}_r&\coloneq&\{\vertexOF{\LL_r^{i}}:0\leq i\leq |\E|\}\mbox{,}&\strut\:\:&
\mathcal{R}_r&\coloneq&\{\vertexOF{\RR_r^{i}}:0\leq i\leq |\E|\}\mbox{.}\\
\end{array}
\end{equation}
For simplicity of notation an $\gadgetE_r$-path $P$ is denoted by $\LL^\ast_r$ if there exists a positive integer $i$ such that
$P=\LL^i_r$, and analogously we denote by $\RR^\ast_r$ an $\gadgetE_r$-path $P$ if there exists a positive integer $j$ such
that $P=\RR^j_r$.
\begin{figure}[ht!]
\begin{center}
\begin{tikzpicture}
[>=stealth,->,shorten >=2pt,looseness=.5,auto]
    \matrix (M)[matrix of math nodes,
        column sep={\g,between origins},
        row sep={\g,between origins},
        nodes=mynodes]
        {
&&&&|(A00)|&&&&\\ &&&|(A01)|&&|(A10)|&&&\\ &&|(A02)|&&|(A11)|&&|(A20)|&&\\ &|(A03)|&&|(A12)|&&|(A21)|&&|(A30)|&\\
|(A04)|&&|(A13)|&&|(A22)|&&|(A31)|&&|(A40)|\\ &|(A14)|&&|(A23)|&&|(A32)|&&|(A41)|&\\ &&|(A24)|&&|(A33)|&&|(A42)|&&\\
&&&|(A34)|&&|(A43)|&&&\\ &&&&|(A44)|&&&&\\
        };
    \tikzstyle{every node}=[font=\small\itshape]
\draw[Referencja](A00)--(A01);\draw[Referencja](A00)--(A10);\draw[Referencja](A10)-- (A11);\draw[Referencja](A10) --
(A20);\draw[Referencja](A01) -- (A11);\draw[Referencja](A01) -- (A02);\draw[Referencja,dotted](A02) --
(A03);\draw[Referencja](A02) -- (A12);\draw[Referencja](A11) -- (A12);\draw[Referencja](A11) -- (A21);\draw[Referencja](A20) --
(A21);\draw[Referencja,dotted](A20) -- (A30);\draw[Referencja](A30) -- (A40);    \draw[Referencja](A30) -- (A31);
\draw[Referencja,dotted](A21) -- (A31);        \draw[Referencja](A21) -- (A22);        \draw[Referencja](A12) -- (A22);
\draw[Referencja,dotted](A12) -- (A13);        \draw[Referencja](A03) -- (A13);        \draw[Referencja](A03) -- (A04);
\draw[Referencja](A04) -- (A14);        \draw[Referencja](A13) -- (A14);        \draw[Referencja](A13) -- (A23);
\draw[Referencja,dotted](A22) -- (A23);       \draw[Referencja,dotted](A22) -- (A32);\draw[Referencja](A31) --
(A32);\draw[Referencja](A31) -- (A41);\draw[Referencja](A40) -- (A41);\draw[Referencja](A40) -- (A41);\draw[Referencja](A41) --
(A42);\draw[Referencja](A32) -- (A42);\draw[Referencja,dotted](A32) -- (A33);\draw[Referencja,dotted](A23) --
(A33);\draw[Referencja](A23) -- (A24);\draw[Referencja](A14) -- (A24);\draw[Referencja,dotted](A24) --
(A34);\draw[Referencja](A33) -- (A34);\draw[Referencja](A33) -- (A43);\draw[Referencja,dotted](A42) --
(A43);\draw[Referencja,-,dotted](A22) -- (A33);\draw[Referencja](A43) -- (A44);\draw[Referencja](A34) -- (A44);
\node at (A00) [right,inner sep=2mm]{$r_{0,0}$};
\node at (A01) [right,inner sep=2mm]{$r_{0,1}$};
\node at (A10) [right,inner sep=2mm]{$r_{1,0}$};
\node at (A11) [right,inner sep=2mm]{$r_{1,1}$};
\node at (A20) [right,inner sep=2mm]{$r_{2,0}$};
\node at (A02) [right,inner sep=2mm]{$r_{0,2}$};
\node at (A03) [right,inner sep=2mm]{$r_{0,|\E|\!-\!1}$};
\node at (A04) [right,inner sep=2mm]{$r_{0,|\E|}$};
\node at (A13) [right,inner sep=2mm]{$r_{1,|\E|\!-\!1}$};
\node at (A14) [right,inner sep=2mm]{$r_{1,|\E|}$};
\node at (A12) [right,inner sep=2mm]{$r_{1,2}$};
\node at (A21) [right,inner sep=2mm]{$r_{2,1}$};
\node at (A30) [right,inner sep=2mm]{$r_{|\E|\!-\!1,0}$};
\node at (A22) [right,inner sep=2mm]{$r_{2,2}$};
\node at (A31) [right,inner sep=2mm]{$r_{|\E|\!-\!1,1}$};
\node at (A40) [right,inner sep=2mm]{$r_{|\E|,0}$};
\node at (A23) [right,inner sep=2mm]{$r_{2,|\E|\!-\!1}$};
\node at (A32) [right,inner sep=2mm]{$r_{|\E|\!-\!1,2}$};
\node at (A41) [right,inner sep=2mm]{$r_{|\E|,1}$};
\node at (A24) [right,inner sep=2mm]{$r_{2,|\E|}$};
\node at (A33) [right,inner sep=1.6mm]{ $r_{|\!\mkern1mu\E\!\mkern1mu|\!-\!1,|\!\mkern1mu\E\!\mkern1mu|\!-\!1}$};
\node at (A42) [right,inner sep=2mm]{$r_{|\E|,2}$};
\node at (A34) [right,inner sep=2mm]{$r_{|\E|\!-\!1,|\E|}$};
\node at (A43) [right,inner sep=2mm]{$r_{|\E|,|\E|\!-\!1}$};
\node at (A44) [right,inner sep=2mm]{$r_{|\E|,|\E|}$};
\end{tikzpicture}
\vspace{-0.5cm}
\end{center}
\caption{The gadget $\gadget_r$.}
\label{G1}
\end{figure}

\par Obviously $\gadget_r$ is a \DAG{} without $\gadgetE_r$-shortcuts,
since an arbitrary $\gadgetE_r$-arc $\langle r_{i_1,j_1},r_{i_2,j_2}\rangle$ fulfils the condition $i_1\!+\!j_1\!+\!1\! =\!
i_2\!+\!j_2$. Moreover, a set of $\gadgetE_r$-paths that can determine an $\mathcal{H}^{(\ast)}$-partition of
$\gadget_r$ is included only in paths that are subsequences of $\LL^\ast_r$ or $\RR^\ast_r$. {It is expressed more precisely in
the following lemma.}

\begin{lem}\label{R2_S1_WN1}
Let $\pi$ be an $\mathcal{H}^{(\ast)}$-partition of $\gadget_r$ and $P\in \pi$. Then the path $\pipath{\pi}{P}$
has {arcs in $\LL_r$ only or in $\RR_r$ only}.
\end{lem}
\begin{proof}
{The proof is by contradiction.} Suppose that $\pipath{\pi}{P}$ uses at least one $\LL_r$-arc and at least one $\RR_r$-arc.
Note that for each consecutive elements of this path $r_{i_1,j_1}$, $r_{i_2,j_2}$ we have $i_1+j_1<i_2+j_2$. Then from the
assumption about $\LL_r$-arcs and $\RR_r$-arc $G$, we obtain that $\pipath{\pi}{P}$ has $\langle r_{i,j},r_{i,j\!+\!1}\rangle$,
$\langle r_{i,j\!+\!1},r_{i\!+\!1,j\!+\!1}\rangle$ arcs or $\langle r_{i,j},r_{i\!+1,j}\rangle$, $\langle
r_{i\!+1,j},r_{i\!+\!1,j\!+\!1}\rangle$, for some $0\leq i,j<|\E|$. We prove only the first case, the second case is analogous.
Since $i+1+j< i+1+j+1<k+l$ for each $r_{k,l}\in P$ after $r_{i+1,j+1}$ and $i+1+j>i+j>k+l$ for each $r_{k,l}\in P$ before
$r_{i,j}$, we obtain that $\pipath{\pi}{P}\neq\pipath{\pi}{r_{i\!+\!1,j}}$. Then \mbox{$\langle
P,\vertexOF{\pipath{\pi}{r_{i\!+\!1,j}}}\rangle$, $\langle \vertexOF{\pipath{\pi}{r_{i\!+\!1,j}}},P\rangle$} are arcs in
$\pidigraph{\gadget_r^n}{\pi}$, which generates a directed cycle in the acyclic graph $\pidigraph{\gadget_r^n}{\pi}$. {This is
a contradiction with the assumption that $\pi$ is an $\mathcal{H}^{(\ast)}$-partition.}
\end{proof}

{We can now prove that either parts of all $\LL_r^i$ or all $\RR_r^i$ can be found in an $\mathcal{H}
^{(\ast)}$-partition.}

\begin{lem}\label{R2_S1_new_1}
Let $\pi$ be an $\mathcal{H}^{(\ast)}$-partition of $\gadget_r$. Then {for each $i\in\{0,1,\ldots,|\E|\}$} there
exists $L_i\in \pi$ such that $\pipath{\pi}{L_i}$ is a subsequence of $\LL^i_r$ or {for each $i\in\{0,1,\ldots,|\E|\}$} there
exists $R_i\in \pi$ such that $\pipath{\pi}{R_i}$ is a subsequence of $\RR^i_r$.
\end{lem}
\begin{proof}
Suppose that there exists $i$, $0\leq i\leq |\E|$, such that for every $L\in\pi$, $\pipath{\pi}{L}$ is not a subsequence of
$\LL^i_r$. Then from Lemma~\ref{R2_S1_WN1} we conclude that
$\pipath{\pi}{r_{i,0}},\pipath{\pi}{r_{i,1}},\ldots,\pipath{\pi}{r_{i,|\E|}}$ are pairwise different paths and
$\pipath{\pi}{r_{i,j}}$ is a subsequence of $\RR^j_r$, for $j=0,1,\ldots,|\E|$, which completes the proof.
\end{proof}

\par An easy computation shows that from Lemma~\ref{R2_S1_WN1} and~\ref{R2_S1_new_1} we can infer the following lemma.
\begin{lem}\label{R2_S1_T1}
An arbitrary $\mathcal{H} ^{(\ast)}$-partition of $\gadget_r$ has size at least $|\E|+1$. Moreover, there exist
exactly two $\mathcal{H} ^{(|\E|+1)}$-partitions of $\gadget_r$, they are $\mathcal{L}_r$, $\mathcal{R}_r$.
\end{lem}

\subsection{The \DAG{} \texorpdfstring{$\DD{\G}{\e}$}{P(G,e)}}\label{SS2}

In the construction of the digraph $\DD{\G}{\e}$ we use a one-to-one function $\e$, but its choice does not affect the
properties of the considered digraph.

\begin{defi}\label{DAG}
Let $\e:\E\longrightarrow \{1,2,\ldots,|\E|\}$ be a one-to-one function. We define {a graph} $\DD{\G}{\e} \coloneq \langle
\VV,\EE\rangle$ (see Fig.~\ref{G_konstr}) as:
\begin{equation}
\begin{array}{rcl}
\VV &=& \underset{v \in \V}{\bigcup}\,\gadgetV_v \:\cup\:\E,\\
\EE & = &
\underset{v \in \V}{\bigcup}\,\gadgetE_v \:\cup\:
\{\langle \langle v,u\rangle,v_{\ee{\langle v,u\rangle},0}\rangle: \langle v,u\rangle\in \E\}\cup\\
&&\strut\mkern92mu\{\langle \langle v,u\rangle,u_{0,\ee{\langle v,u\rangle}}\rangle: \langle v,u\rangle\in \E\}.
\end{array}\label{KP_eq_14}
\end{equation}
\end{defi}
Clearly, $\DAGcut{\DD{\G}{\e}}{\gadgetV_v} = \gadget_v$ for every $v\in \V$.

\noindent In the following considerations, we construct an $\mathcal{H}^{(\ast)}$-partition $\pi(\F)$ of $\DD{\G}{\e}$ for {a} given feedback arc set $\F$ of $\G$, which fulfils
$|\pi(\F)|=|\V|\cdot(|\E|+1)+|\F|$. Note that throughout this construction the {size of} in-degree and out-degree in $\G$ is
not used.
\begin{figure}[h!]
\begin{center}
\begin{tikzpicture}
\def\c{1.1cm}
\def\r{1.1*1.73cm}
\small
     [>=stealth,->,shorten >=2pt,looseness=.5,auto]
    \matrix (M)[matrix of math nodes,
        column sep={\c,between origins},
        row sep={\r,between origins},
        nodes=mynodes]
{ &\node(W) {w};&\\
\node(V) {v};&&\node(U) {u};\\
&&\\ };
\draw[Referencja,bend left,looseness=0.9](W)to(V);
\draw[Referencja](V)--(U);
\draw[Referencja,bend left,looseness=0.9](V)to(W);
\draw[Referencja](U)--(W);
\node at (W)[left,inner sep=11mm]{$\G$};
\end{tikzpicture}$\strut\mkern-80mu$
\begin{tikzpicture}
\def\c{1.3cm}
\def\r{0.9cm}
\small
     [>=stealth,->,shorten >=2pt,looseness=.5,auto]
    \matrix (M)[matrix of math nodes,
        column sep={\c,between origins},
        row sep={\r,between origins},
        nodes={mynodes,rectangle, rounded corners}]
{ &\node(B0) {(\langle w,v\rangle)};&\\
\node(A1) {\strut\;\LL_v^{0}\;};&
\node(B1) {\strut\;\LL_u^{0}\;};&
\node(C1) {\strut\;\LL_w^{0}\;};\\
\node(A2) {\LL^{\langle v,u \rangle}};&
\node(B2) {\strut\;\LL_u^{1}\;};&
\node(C2) {\strut\;\LL_w^{1}\;};\\
\node(A3) {\strut\;\LL_v^{2}\;};&
\node(B3) {\LL_u^{\langle u,w \rangle}};&
\node(C3) {\strut\;\LL_w^{2}\;};\\
\node(A4) {\strut\;\LL_v^{3}\;};&
\node(B4) {\strut\;\LL_u^{3}\;};&
\node(C4) {\strut\;\LL_w^{3}\;};\\
\node(A5) {\LL^{\langle v,w \rangle}};&
\node(B5) {\strut\;\LL_u^{4}\;};&
\node(C5) {\strut\;\LL_w^{4}\;};\\
};
\draw[Referencja](A1)--(A2);
\draw[Referencja](A2)--(A3);
\draw[Referencja](A3)--(A4);
\draw[Referencja](A4)--(A5);
\draw[Referencja](B1)--(B2);
\draw[Referencja](B2)--(B3);
\draw[Referencja](B3)--(B4);
\draw[Referencja](B4)--(B5);
\draw[Referencja](C1)--(C2);
\draw[Referencja](C2)--(C3);
\draw[Referencja](C3)--(C4);
\draw[Referencja](C4)--(C5);
\draw[Referencja](A2)--(B1);
\draw[Referencja]($0.7*(B3.north east)+0.3*(B3.north)$)--
($0.7*(C1.south west)+0.3*(C1.west)$);
\draw[Referencja](B0).. controls ($(B0)+(\c*2,0)$) and ($(C4)+(\c,1.3*\r)$).. (C4);
\draw[Referencja](B0).. controls ($(B0)+(-\c*2,0)$) and ($(A4)+(-\c,1.3*\r)$).. (A4);
\draw[Referencja]($(A5.south east)$).. controls ($(A5)+(\c*3,-1.4*\r)$) and ($(C1)+(1*\c,-4*\r)$).. ($(C1.south east)$);
\node at (B0)[left,inner sep=19mm]{$\pidigraph{\DD{\G}{\e}}{\pi({\{\langle w,v\rangle\}})}$};
\end{tikzpicture}\\
\vspace{-1cm}
\begin{tikzpicture}
\def\s{0.4cm}
\def\p{0.125cm}
\tiny
[>=stealth,->,shorten >=2pt,looseness=.5,auto]
    \matrix (M)[matrix of math nodes,
        column sep={\s,between origins},
        row sep={\s,between origins},
        nodes=mynodes]
{ &&&&|(A00)|&&&&                              && &&&&|(B00)|&&&&                              && &&&&|(C00)|&&&&
\\ &&&|(A01)|&&|(A10)|&&&                       && &&&|(B01)|&&|(B10)|&&&                       && &&&|(C01)|&&|(C10)|&&&
\\ &&|(A02)|&&|(A11)|&&|(A20)|&&                && &&|(B02)|&&|(B11)|&&|(B20)|&&                &&
&&|(C02)|&&|(C11)|&&|(C20)|&&    \\ &|(A03)|&&|(A12)|&&|(A21)|&&|(A30)|&         &|(vw)|& &|(B03)|&&|(B12)|&&|(B21)|&&|(B30)|&
&& &|(C03)|&&|(C12)|&&|(C21)|&&|(C30)|&\\
|(A04)|&&|(A13)|&&|(A22)|&&|(A31)|&&|(A40)|&&|(B04)|&&|(B13)|&&|(B22)|&&|(B31)|&&|(B40)|&&|(C04)|&&|(C13)|&&|(C22)|&&|(C31)|&&|(C40)|\\
&|(A14)|&&|(A23)|&&|(A32)|&&|(A41)|&         && &|(B14)|&&|(B23)|&&|(B32)|&&|(B41)|&         &&
&|(C14)|&&|(C23)|&&|(C32)|&&|(C41)|&\\ &&|(A24)|&&|(A33)|&&|(A42)|&&                && &&|(B24)|&&|(B33)|&&|(B42)|&&
&& &&|(C24)|&&|(C33)|&&|(C42)|&&\\ &&&|(A34)|&&|(A43)|&&&                       && &&&|(B34)|&&|(B43)|&&&
&& &&&|(C34)|&&|(C43)|&&&\\ &&&&|(A44)|&&&&                              && &&&&|(B44)|&&&&                              &&
&&&&|(C44)|&&&&                \\ };
\draw[-,rounded corners=2mm,shorten >=2pt]
    ($(A00)+(-\p,\p)$)--++(-4*\s-\p,-\p-4*\s)--++(\p+\p,-\p-\p)--++(4*\s+2*\p,4*\s+2*\p)--++(-2*\p,2*\p)--++(-\s,-\s);
\draw[-,rounded corners=2mm,shorten >=2pt]
    ($(A20)+(-\p,\p)$)--++(-4*\s-\p,-\p-4*\s)--++(\p+\p,-\p-\p)--++(4*\s+2*\p,4*\s+2*\p)--++(-2*\p,2*\p)--++(-\s,-\s);
\draw[-,rounded corners=2mm,shorten >=2pt]
    ($(A30)+(-\p,\p)$)--++(-4*\s-\p,-\p-4*\s)--++(\p+\p,-\p-\p)--++(4*\s+2*\p,4*\s+2*\p)--++(-2*\p,2*\p)--++(-\s,-\s);
\draw[-,rounded corners=2mm,shorten >=2pt]
    ($(vw)+(-\p,\p)$)--++(-5*\s-\p,-\p-5*\s)--++(\p+\p,-\p-\p)--++(5*\s+2*\p,5*\s+2*\p)--++(-2*\p,2*\p)--++(-\s,-\s);
\draw[-,rounded corners=2mm,shorten >=2pt]
    ($(B00)+(-\p,\p)$)--++(-4*\s-\p,-\p-4*\s)--++(\p+\p,-\p-\p)--++(4*\s+2*\p,4*\s+2*\p)--++(-2*\p,2*\p)--++(-\s,-\s);
\draw[-,rounded corners=2mm,shorten >=2pt]
    ($(B10)+(-\p,\p)$)--++(-4*\s-\p,-\p-4*\s)--++(\p+\p,-\p-\p)--++(4*\s+2*\p,4*\s+2*\p)--++(-2*\p,2*\p)--++(-\s,-\s);
\draw[-,rounded corners=2mm,shorten >=2pt]
    ($(B30)+(-\p,\p)$)--++(-4*\s-\p,-\p-4*\s)--++(\p+\p,-\p-\p)--++(4*\s+2*\p,4*\s+2*\p)--++(-2*\p,2*\p)--++(-\s,-\s);
\draw[-,rounded corners=2mm,shorten >=2pt]
    ($(B40)+(-\p,\p)$)--++(-4*\s-\p,-\p-4*\s)--++(\p+\p,-\p-\p)--++(4*\s+2*\p,4*\s+2*\p)--++(-2*\p,2*\p)--++(-\s,-\s);
\draw[-,rounded corners=2mm,shorten >=2pt]
    ($(C00)+(-\p,\p)$)--++(-4*\s-\p,-\p-4*\s)--++(\p+\p,-\p-\p)--++(4*\s+2*\p,4*\s+2*\p)--++(-2*\p,2*\p)--++(-\s,-\s);
\draw[-,rounded corners=2mm,shorten >=2pt]
    ($(C10)+(-\p,\p)$)--++(-4*\s-\p,-\p-4*\s)--++(\p+\p,-\p-\p)--++(4*\s+2*\p,4*\s+2*\p)--++(-2*\p,2*\p)--++(-\s,-\s);
\draw[-,rounded corners=2mm,shorten >=2pt]
    ($(C20)+(-\p,\p)$)--++(-4*\s-\p,-\p-4*\s)--++(\p+\p,-\p-\p)--++(4*\s+2*\p,4*\s+2*\p)--++(-2*\p,2*\p)--++(-\s,-\s);
\draw[-,rounded corners=2mm,shorten >=2pt]
    ($(C30)+(-\p,\p)$)--++(-4*\s-\p,-\p-4*\s)--++(\p+\p,-\p-\p)--++(4*\s+2*\p,4*\s+2*\p)--++(-2*\p,2*\p)--++(-\s,-\s);
\draw[-,rounded corners=2mm,shorten >=2pt]
    ($(C40)+(-\p,\p)$)--++(-4*\s-\p,-\p-4*\s)--++(\p+\p,-\p-\p)--++(4*\s+2*\p,4*\s+2*\p)--++(-2*\p,2*\p)--++(-\s,-\s);
\draw[Referencja](A00)--(A01);
\draw[Referencja](A00)--(A10);
\draw[Referencja](A10)--(A11);
\draw[Referencja](A10)--(A20);
\draw[Referencja](A01)--(A02);
\draw[Referencja](A01)--(A11);
\draw[Referencja](A02)--(A03);
\draw[Referencja](A02)--(A12);
\draw[Referencja](A11)--(A12);
\draw[Referencja](A11)--(A21);
\draw[Referencja](A20)--(A21);
\draw[Referencja](A20)--(A30);
\draw[Referencja](A30)--(A40);
\draw[Referencja](A30)--(A31);
\draw[Referencja](A21)--(A31);
\draw[Referencja](A21)--(A22);
\draw[Referencja](A12)--(A22);
\draw[Referencja](A12)--(A13);
\draw[Referencja](A03)--(A13);
\draw[Referencja](A03)--(A04);
\draw[Referencja](A04)--(A14);
\draw[Referencja](A13)--(A14);
\draw[Referencja](A13)--(A23);
\draw[Referencja](A22)--(A23);
\draw[Referencja](A22)--(A32);
\draw[Referencja](A31)--(A32);
\draw[Referencja](A31)--(A41);
\draw[Referencja](A40)--(A41);
\draw[Referencja](A14)--(A24);
\draw[Referencja](A23)--(A24);
\draw[Referencja](A23)--(A33);
\draw[Referencja](A32)--(A33);
\draw[Referencja](A32)--(A42);
\draw[Referencja](A41)--(A42);
\draw[Referencja](A24)--(A34);
\draw[Referencja](A33)--(A34);
\draw[Referencja](A33)--(A43);
\draw[Referencja](A42)--(A43);
\draw[Referencja](A34)--(A44);
\draw[Referencja](A43)--(A44);
\draw[Referencja](B00)--(B01);
\draw[Referencja](B00)--(B10);
\draw[Referencja](B10)--(B11);
\draw[Referencja](B10)--(B20);
\draw[Referencja](B01)--(B02);
\draw[Referencja](B01)--(B11);
\draw[Referencja](B02)--(B03);
\draw[Referencja](B02)--(B12);
\draw[Referencja](B11)--(B12);
\draw[Referencja](B11)--(B21);
\draw[Referencja](B20)--(B21);
\draw[Referencja](B20)--(B30);
\draw[Referencja](B30)--(B40);
\draw[Referencja](B30)--(B31);
\draw[Referencja](B21)--(B31);
\draw[Referencja](B21)--(B22);
\draw[Referencja](B12)--(B22);
\draw[Referencja](B12)--(B13);
\draw[Referencja](B03)--(B13);
\draw[Referencja](B03)--(B04);
\draw[Referencja](B04)--(B14);
\draw[Referencja](B13)--(B14);
\draw[Referencja](B13)--(B23);
\draw[Referencja](B22)--(B23);
\draw[Referencja](B22)--(B32);
\draw[Referencja](B31)--(B32);
\draw[Referencja](B31)--(B41);
\draw[Referencja](B40)--(B41);
\draw[Referencja](B14)--(B24);
\draw[Referencja](B23)--(B24);
\draw[Referencja](B23)--(B33);
\draw[Referencja](B32)--(B33);
\draw[Referencja](B32)--(B42);
\draw[Referencja](B41)--(B42);
\draw[Referencja](B24)--(B34);
\draw[Referencja](B33)--(B34);
\draw[Referencja](B33)--(B43);
\draw[Referencja](B42)--(B43);
\draw[Referencja](B34)--(B44);
\draw[Referencja](B43)--(B44);
\draw[Referencja](C00)--(C01);
\draw[Referencja](C00)--(C10);
\draw[Referencja](C10)--(C11);
\draw[Referencja](C10)--(C20);
\draw[Referencja](C01)--(C02);
\draw[Referencja](C01)--(C11);
\draw[Referencja](C02)--(C03);
\draw[Referencja](C02)--(C12);
\draw[Referencja](C11)--(C12);
\draw[Referencja](C11)--(C21);
\draw[Referencja](C20)--(C21);
\draw[Referencja](C20)--(C30);
\draw[Referencja](C30)--(C40);
\draw[Referencja](C30)--(C31);
\draw[Referencja](C21)--(C31);
\draw[Referencja](C21)--(C22);
\draw[Referencja](C12)--(C22);
\draw[Referencja](C12)--(C13);
\draw[Referencja](C03)--(C13);
\draw[Referencja](C03)--(C04);
\draw[Referencja](C04)--(C14);
\draw[Referencja](C13)--(C14);
\draw[Referencja](C13)--(C23);
\draw[Referencja](C22)--(C23);
\draw[Referencja](C22)--(C32);
\draw[Referencja](C31)--(C32);
\draw[Referencja](C31)--(C41);
\draw[Referencja](C40)--(C41);
\draw[Referencja](C14)--(C24);
\draw[Referencja](C23)--(C24);
\draw[Referencja](C23)--(C33);
\draw[Referencja](C32)--(C33);
\draw[Referencja](C32)--(C42);
\draw[Referencja](C41)--(C42);
\draw[Referencja](C24)--(C34);
\draw[Referencja](C33)--(C34);
\draw[Referencja](C33)--(C43);
\draw[Referencja](C42)--(C43);
\draw[Referencja](C34)--(C44);
\draw[Referencja](C43)--(C44);
\draw[Referencja](vw)--(A40);
\node(vu)[mynodes] at ($(A00)+(5*\s,\s)$){};
\draw[-,rounded corners=2mm,shorten >=2pt]
    ($(A10)+(-\p,\p)$)--++(-4*\s-\p,-\p-4*\s)--++(\p+\p,-\p-\p)--++(4*\s+2*\p,4*\s+2*\p);
\draw[-]
    ($(A00)+(-\p+\s,\p-\s)$) .. controls ($(A00)+(-\p+\s,\p-\s)+(0.41*\p+2*\s,0.41*\p+2*\s)$)
    .. ($(vu)+(0,1.41*\p)$);
\draw[-]
    ($(A00)+(\p+\s,-\p-\s)$) .. controls ($(A00)+(\p+\s,-\p-\s)+(-0.41*\p+2*\s,-0.41*\p+2*\s)$)
    .. ($(vu)+(0,-1.41*\p)$);
\draw[Referencja]
($(vu)+(-0.139,0)$) .. controls ($(A00)+(+\s,-\s)+(-0.41*\p+2*\s,-0.41*\p+2*\s)$)
    .. ($(A00)+(+\s,-\s)+(0.1,0.1)$);
\draw[Referencja]
($(vu)+(0.139,0)$) .. controls ($(A00)+(5.5*\s,-\s)+(-0.41*\p+2*\s,-0.41*\p+2*\s)$)
    .. ($(A00)+(9*\s,-\s)+(-0.1,0.1)$);
\draw[-,rounded corners=2mm,shorten >=2pt]
    ($(vu)+(0,1.41*\p)$)--++(1.41*\p,0)--++(0,2*-1.41*\p)--++(-1.61*\p,0);
\node(uw)[mynodes] at ($(B00)+(5*\s,\s)$){};
\draw[-,rounded corners=2mm,shorten >=2pt]
    ($(uw)+(-\p,\p)$)--++(-7*\s-\p,-\p-7*\s)--++(\p+\p,-\p-\p)--++(7*\s+2*\p,7*\s+2*\p)--++(-2*\p,2*\p)--++(-\s,-\s);
\draw[Referencja](uw)--(C02);
\draw[Referencja](uw)--(B20);
\node(wv)[mynodes] at ($(B00)+(0,3*\s)$){};
\draw[Referencja](wv)..controls ($(wv)+(11*\s,0)$) and ($(C30)+(\s,4*\s)$)..(C30);
\draw[Referencja](wv)..controls ($(wv)+(-11*\s,0)$) and ($(A03)+(-\s,4*\s)$)..(A03);

\draw[Referencja,-](vw)..controls ($(vw)+(2*\s,-6*\s)$)and ($(B44)+(-\s,-\s)$).. ($(B44)+(0,-\s)$);
\draw[Referencja]($(B44)+(0,-\s)$) .. controls ($(B44)+(\s,-\s)$)
..(C04);
\draw[-,rounded corners=1.75mm,shorten >=2pt]
($(wv)+(-\p,\p)$)--++(-\p,-\p)--++(2*\p,-2*\p)--++(2*\p,2*\p)--++(-2*\p,2*\p)--++(-\p,-\p);
\small
\node at (wv) [above,inner sep=2mm]{$\langle w,v\rangle$};
\node at (uw) [above,inner sep=2mm]{$\langle u,w\rangle$};
\node at (vu) [above,inner sep=2mm]{$\langle v,u\rangle$};
\node at (vw) [above,inner sep=2mm]{$\langle v,w\rangle$};
\node at (A03) [above left,inner sep=2mm]{$v_{0,3}$};
\node at (A10) [right,inner sep=4mm]{$v_{1,0}$};
\node at (A40) [below right,inner sep=2mm]{$v_{4,0}$};
\node at (B01) [left,inner sep=3mm]{$u_{0,1}$};
\node at (B20) [right,inner sep=4mm]{$u_{2,0}$};
\node at (C02) [left,inner sep=3mm]{$w_{0,2}$};
\node at (C30) [right,inner sep=3mm]{$w_{3,0}$};
\node at ($(wv)+(-5,0.2)$){$\DD{\G}{\e}$};
\node at (C04) [above left,inner sep=1mm]{$w_{0,4}$};
\end{tikzpicture}
\end{center}
\caption{The $\mathcal{H}^{(16)}$-partition $\pi(\F)$ of
$\DD{\G}{\e}$, an example that illustrates the construction from the proof of Lemma~\ref{R2_S1_prosta_implikacja}, where
$\ee{\langle v,u\rangle}=1,$ $\ee{\langle u,w\rangle}=2,$ $\ee{\langle w,v\rangle}=3,$ $\ee{\langle v,w\rangle}=4,$
$\F=\{\langle w,v\rangle\}$.}
\label{G_konstr}
\end{figure}
\begin{lem}\label{R2_S1_prosta_implikacja}
Let $\F$ be a feedback arc set of $\G$. Then there exists an $\mathcal{H}^{(|\V|\cdot(|\E|+1)+|\F|)}$-partition of
$\DD{\G}{\e}$.
\end{lem}
\begin{proof}
 {We define} a partition $\pi(\F)$ of $\VV$ determined by $\F$ {as:}
\begin{equation}
\begin{array}{rl}
\pi(\F)\mkern-1mu =\mkern-9mu&\{\{\eT{i},v_{i,0},v_{i,1},\ldots,v_{i,|\E|}\}:\\
&\multicolumn{1}{r}{ v \in \V \wedge 0\leq i\leq |\E|
\wedge\! \underset{u \in \V}{\exists}
(\eT{i}=\langle v,u\rangle\wedge \langle v,u\rangle\! \in\! \E\!\setminus\! \F)\}\,\cup}\\
&\{\{v_{i,0},v_{i,1},\ldots,v_{i,|\E|}\}:\\ &\multicolumn{1}{r}{\strut\mkern80mu v \in \V \wedge 0\leq i\leq |\E|
\wedge \neg\underset{u \in \V}{\exists}
(\eT{i}=\langle v,u\rangle\wedge \langle v,u\rangle \in \E\setminus \F)\}\,\cup}\\ &\{\{\langle v,u\rangle\}:\langle v,u\rangle
\in  \F\}.\\
\end{array}
\end{equation}
Elements of $\pi(\F)$ can be divided into three types of $\VV$-paths:
\begin{enumerate}
\item\label{type_one} $\langle\langle v,u\rangle\rangle\!\mathbin{{}^\frown}\!$\!$\LL_v^{\ee{\langle v,u\rangle}}$, if
$\langle v,u\rangle\in \E\setminus \F$ (denoted by $\LL^{\langle v,u\rangle}$),
\item\label{type_two} $\LL_v^{i}$, if $\neg\underset{u \in \V}{\exists}
(\eT{i}=\langle v,u\rangle\wedge \langle v,u\rangle \in \E\setminus \F)$,
\item\label{type_three} $\langle\langle v,u\rangle\rangle$, if $\langle v,u\rangle \in  \F$.
\end{enumerate}
It is easy to check that $\pi(\F)$ is an $\mathcal{H} ^{|\V|\cdot(|\E|+1)+|\F|}$-partition of $\DD{\G}{\e}$. To complete
the proof it is enough to show only that $\pidigraph{\DD{\G}{\e}}{\pi(\F)}$ is acyclic.

\par Suppose, contrary to our claim, that there exists a cycle $\mathfrak{a}=\langle a_1,a_2,\ldots,a_k\rangle$ {in}\break
$\pidigraph{\DD{\G}{\e}}{\pi(\F)}$, where $k> 2$ and $a_k=a_1$.
Since all paths of the \ref{type_three}rd type correspond to a source in $\pidigraph{\DD{\G}{\e}}{\pi(\F)}$, see
\eqref{KP_eq_14}, we conclude that for each $i\in \{1,2,\ldots,k-1\}$ there is a vertex $v$ such that $a_i\cap
\gadgetV_{v}\neq\emptyset$. Moreover this vertex is unique for each $i$. {Let $v_i$ be such vertex corresponding to $a_i$.}
Obviously, vertices $v_1,v_2,\ldots,v_{k-1}$ do not have to be pairwise different. {However let us consider the maximal
subsequence $\langle v_1^\prime, v_2^\prime,\ldots,v_{k^\prime-1}^\prime,v_{k^\prime}^\prime\rangle$ of $\langle
v_1,v_2,\ldots,v_{k-1},v_1\rangle$ such that $v_i^\prime\neq v_{i+1}^\prime$ for $i=1,2,\ldots,k^\prime-1$ where
$v^\prime_k=v_1^\prime$. We prove that this subsequence is} a cycle of {the} graph $\langle \V,\E\setminus\F\rangle$ which
contradicts the fact $\F$ is a feedback arc set of $\G$.

\par For this purpose let us consider vertices $a_i$, $a_{i+1}$ of $\pidigraph{\DD{\G}{\e}}{\pi(\F)}$ which satisfy $v_i\neq
v_{i+1}$.
Obviously, both vertices $a_i$, $a_{i+1}$ have to correspond to $\VV$-paths of \ref{type_one}st or \ref{type_two}nd type. From
this, it is easy to check that since $\langle a_i,a_{i+1}\rangle$ is an arc of $\pidigraph{\DD{\G}{\e}}{\pi(\F)}$, we have
$a_i=\mathcal{V}(\LL^{\langle v_i,v_{i+1}\rangle})$, $a_{i+1}=\mathcal{V}(\LL_{v_{i+1}}^{0})$, and the form of $a_i$ implies
that $\langle v_i,v_{i+1}\rangle\in \E\setminus\F$.
\end{proof}

\subsection{A feedback arc set of \texorpdfstring{$\G$}{G}
determined by an \texorpdfstring{$\mathcal{H} ^{(\ast)}$}{H*}-partition of \texorpdfstring{$\DD{\G}{\e}$}{P(G,e}}\label{SS3}

Let us consider an $\mathcal{H} ^{(\ast)}$-partition $\pi$ of $\DD{\G}{\e}$. First we show in Lemma~\ref{modyf} that
there exists an $\mathcal{H}^{(\ast)}$-partition $\pi^\prime$ of $\DD{\G}{\e}$ which satisfies $|\pi^\prime|\leq |\pi|$
and additionally has the property that {for each $v\in \V$}: $\pi^\prime$ restricted to the gadget $\gadget_v$ is a
$\mathcal{L}_v$ or $\mathcal{R}_v$-partition. Since every gadget $\gadget_v$ in $\pi^\prime$ is a $\mathcal{L}_v$ or
$\mathcal{R}_v$-partition, we construct in Lemma~\ref{construction} a feedback arc set of $\G$ which has size at most
$|\pi|-|\V|\cdot(|\E|+1)$.

\par Let $\pi$ be an $\mathcal{H}
^{(\ast)}$--partition of $\DD{\G}{\e}$. For a vertex $v\in \V$ we use the following notations:
$\picut{\pi}{v}\coloneq\{\vertexOF{\pipath{\pi}{w}}:w \in\gadgetV_v\}$,
$\picut{\pi}{\gadget_v}\coloneq\{\vertexOF{\pipath{\pi}{w}}\cap \gadgetV_v:w \in\gadgetV_v\}$,
$\DD{\G}{\e}_{|\pi,v}\coloneq\DAGcut{\DD{\G}{\e}}{{\bigcup}\picut{\pi}{v}}$. {Obviously $\picut{\pi}{\gadget_v}$ is an
$\mathcal{H} ^{(\ast)}$--partition of $\gadget_v$.} We also denote by $\VV_{|\pi,v}$ the set of vertices of
$\DD{\G}{\e}_{|\pi,v}$, and by $\EE_{|\pi,v}$ the set of arcs of $\DD{\G}{\e}_{|\pi,v}$.

\begin{defi}
Let $\pi$ be an $\mathcal{H}^{(\ast)}$--partition of $\DD{\G}{\e}$ and $v\in \V$. We call a gadget $\gadget_v$ {\em
well-oriented in $\pi$} ({see Fig.~\ref{Orien}}) if and only if $\picut{\pi}{\gadget_v}=\mathcal{L}_v$ or
$\picut{\pi}{\gadget_v}=\mathcal{R}_v$. In the case $\picut{\pi}{\gadget_v}=\mathcal{L}_v$ we say that the gadget $\gadget_v$
is {\em $\LL$-oriented in $\pi$}, and {\em $\RR$-oriented in $\pi$}, in the second case.
\end{defi}

{\begin{figure}[h!]
\begin{tikzpicture}
\def\s{0.4cm}
\def\p{0.125cm}
\tiny
[>=stealth,->,shorten >=2pt,looseness=.5,auto]
    \matrix (M)[matrix of math nodes,
        column sep={\s,between origins},
        row sep={\s,between origins},
        nodes=mynodes]
{ &&&&|(A00)|&&&&                              &&&&&& &&&&|(B00)|&&&&                                                 \\
&&&|(A01)|&&|(A10)|&&&                       &&&&&& &&&|(B01)|&&|(B10)|&&&                                   \\
&&|(A02)|&&|(A11)|&&|(A20)|&&                &&&&&& &&|(B02)|&&|(B11)|&&|(B20)|&&             \\
&|(A03)|&&|(A12)|&&|(A21)|&&|(A30)|&         &&&&&& &|(B03)|&&|(B12)|&&|(B21)|&&|(B30)|&      \\
|(A04)|&&|(A13)|&&|(A22)|&&|(A31)|&&|(A40)|&&&&&&|(B04)|&&|(B13)|&&|(B22)|&&|(B31)|&&|(B40)|&\\
&|(A14)|&&|(A23)|&&|(A32)|&&|(A41)|&         &&&&&& &|(B14)|&&|(B23)|&&|(B32)|&&|(B41)|&      \\ &&|(A24)|&&|(A33)|&&|(A42)|&&
&&&&&& &&|(B24)|&&|(B33)|&&|(B42)|&&               \\ &&&|(A34)|&&|(A43)|&&&                       &&&&&&
&&&|(B34)|&&|(B43)|&&&                       \\ &&&&|(A44)|&&&&                              &&&&&& &&&&|(B44)|&&&&
\\ };
\draw[-,rounded corners=2mm,shorten >=2pt]
    ($(A00)+(-\p,\p)$)--++(-4*\s-\p,-\p-4*\s)--++(\p+\p,-\p-\p)--++(4*\s+2*\p,4*\s+2*\p)--++(-2*\p,2*\p)--++(-\s,-\s);
\draw[-,rounded corners=2mm,shorten >=2pt]
    ($(A20)+(-\p,\p)$)--++(-4*\s-\p,-\p-4*\s)--++(\p+\p,-\p-\p)--++(4*\s+2*\p,4*\s+2*\p)--++(-2*\p,2*\p)--++(-\s,-\s);
\draw[-,rounded corners=2mm,shorten >=2pt]
    ($(A30)+(-\p,\p)$)--++(-4*\s-\p,-\p-4*\s)--++(\p+\p,-\p-\p)--++(4*\s+2*\p,4*\s+2*\p)--++(-2*\p,2*\p)--++(-\s,-\s);
\draw[Referencja](A00)--(A01);
\draw[Referencja](A00)--(A10);
\draw[Referencja](A10)--(A11);
\draw[Referencja](A10)--(A20);
\draw[Referencja](A01)--(A02);
\draw[Referencja](A01)--(A11);
\draw[Referencja](A02)--(A03);
\draw[Referencja](A02)--(A12);
\draw[Referencja](A11)--(A12);
\draw[Referencja](A11)--(A21);
\draw[Referencja](A20)--(A21);
\draw[Referencja](A20)--(A30);
\draw[Referencja](A30)--(A40);
\draw[Referencja](A30)--(A31);
\draw[Referencja](A21)--(A31);
\draw[Referencja](A21)--(A22);
\draw[Referencja](A12)--(A22);
\draw[Referencja](A12)--(A13);
\draw[Referencja](A03)--(A13);
\draw[Referencja](A03)--(A04);
\draw[Referencja](A04)--(A14);
\draw[Referencja](A13)--(A14);
\draw[Referencja](A13)--(A23);
\draw[Referencja](A22)--(A23);
\draw[Referencja](A22)--(A32);
\draw[Referencja](A31)--(A32);
\draw[Referencja](A31)--(A41);
\draw[Referencja](A40)--(A41);
\draw[Referencja](A14)--(A24);
\draw[Referencja](A23)--(A24);
\draw[Referencja](A23)--(A33);
\draw[Referencja](A32)--(A33);
\draw[Referencja](A32)--(A42);
\draw[Referencja](A41)--(A42);
\draw[Referencja](A24)--(A34);
\draw[Referencja](A33)--(A34);
\draw[Referencja](A33)--(A43);
\draw[Referencja](A42)--(A43);
\draw[Referencja](A34)--(A44);
\draw[Referencja](A43)--(A44);
\draw[Referencja](B00)--(B01);
\draw[Referencja](B00)--(B10);
\draw[Referencja](B10)--(B11);
\draw[Referencja](B10)--(B20);
\draw[Referencja](B01)--(B02);
\draw[Referencja](B01)--(B11);
\draw[Referencja](B02)--(B03);
\draw[Referencja](B02)--(B12);
\draw[Referencja](B11)--(B12);
\draw[Referencja](B11)--(B21);
\draw[Referencja](B20)--(B21);
\draw[Referencja](B20)--(B30);
\draw[Referencja](B30)--(B40);
\draw[Referencja](B30)--(B31);
\draw[Referencja](B21)--(B31);
\draw[Referencja](B21)--(B22);
\draw[Referencja](B12)--(B22);
\draw[Referencja](B12)--(B13);
\draw[Referencja](B03)--(B13);
\draw[Referencja](B03)--(B04);
\draw[Referencja](B04)--(B14);
\draw[Referencja](B13)--(B14);
\draw[Referencja](B13)--(B23);
\draw[Referencja](B22)--(B23);
\draw[Referencja](B22)--(B32);
\draw[Referencja](B31)--(B32);
\draw[Referencja](B31)--(B41);
\draw[Referencja](B40)--(B41);
\draw[Referencja](B14)--(B24);
\draw[Referencja](B23)--(B24);
\draw[Referencja](B23)--(B33);
\draw[Referencja](B32)--(B33);
\draw[Referencja](B32)--(B42);
\draw[Referencja](B41)--(B42);
\draw[Referencja](B24)--(B34);
\draw[Referencja](B33)--(B34);
\draw[Referencja](B33)--(B43);
\draw[Referencja](B42)--(B43);
\draw[Referencja](B34)--(B44);
\draw[Referencja](B43)--(B44);
\draw[Referencja,dotted]($(A40)+1.2*(\s,\s)$)--(A40);
\draw[Referencja,dotted]($(A10)+1.2*(\s,\s)$)--(A10);
\draw[Referencja,dotted]($(A03)+0.7*(-\s,2.5*\s)$)--(A03);
\draw[-,rounded corners=2mm,shorten >=2pt,dotted]
    ($(A10)+(-\p,\p)+0.95*(\s,\s)$)
    --++(-4*\s-\p-0.95*\s,
         -\p-4*\s-0.95*\s)--++(\p+\p,-\p-\p)
    --++(4*\s+\p+0.95*\s+2pt,
         \p+4*\s+0.95*\s+2pt);
\draw[-,rounded corners=2mm,shorten >=2pt,dotted]
    ($(A40)+(-\p,\p)+0.95*(\s,\s)$)
    --++(-4*\s-\p-0.95*\s,
         -\p-4*\s-0.95*\s)--++(\p+\p,-\p-\p)
    --++(4*\s+\p+0.95*\s+2pt,
         \p+4*\s+0.95*\s+2pt);
\draw[-,rounded corners=2mm,shorten >=2pt]
    ($(B00)+(-\p,\p)$)--++(-2*\s-\p,-\p-2*\s)--++(\p+\p,-\p-\p)
                      --++(2*\s+2*\p,2*\s+2*\p)--++(-2*\p,2*\p)--++(-\s,-\s);
\draw[-,rounded corners=2mm,shorten >=2pt]
    ($(B11)+(-\p,\p)$)--++(-1*\s-\p,-\p-1*\s)--++(\p+\p,-\p-\p)
                      --++(1*\s+2*\p,1*\s+2*\p)--++(-2*\p,2*\p)--++(-\s,-\s);
\draw[-,rounded corners=2mm,shorten >=2pt]
    ($(B10)+(-\p,-\p)$)--++(2*\s+\p,-\p-2*\s)--++(\p+\p,\p+\p)
                      --++(-2*\s-2*\p,2*\s+2*\p)--++(-2*\p,-2*\p)--++(\s,-\s);
\draw[-,rounded corners=2mm,shorten >=2pt]
    ($(B03)+(-\p,-\p)$)--++(3*\s+\p,-\p-3*\s)--++(\p+\p,\p+\p)
                      --++(-3*\s-2*\p,3*\s+2*\p)--++(-2*\p,-2*\p)--++(\s,-\s);
\draw[-,rounded corners=2mm,shorten >=2pt]
    ($(B14)+(-\p,-\p)$)--++(2*\s+\p,-\p-2*\s)--++(\p+\p,\p+\p)
                      --++(-2*\s-2*\p,2*\s+2*\p)--++(-2*\p,-2*\p)--++(\s,-\s);
\draw[-,rounded corners=1.75mm,shorten >=2pt]
($(B04)+(-\p,\p)$)--++(-\p,-\p)--++(2*\p,-2*\p)--++(2*\p,2*\p)--++(-2*\p,2*\p)--++(-\p,-\p);
\draw[-,rounded corners=2mm,shorten >=2pt]
    ($(B22)+(-\p,-\p)$)--++(1*\s+\p,-\p-1*\s)--++(\p+\p,\p+\p)
                      --++(-1*\s-2*\p,1*\s+2*\p)--++(-2*\p,-2*\p)--++(\s,-\s);
\draw[-,rounded corners=2mm,shorten >=2pt]
    ($(B21)+(-\p,-\p)$)--++(1*\s+\p,-\p-1*\s)--++(\p+\p,\p+\p)
                      --++(-1*\s-2*\p,1*\s+2*\p)--++(-2*\p,-2*\p)--++(\s,-\s);
\draw[-,rounded corners=2mm,shorten >=2pt]
    ($(B40)+(-\p,\p)$)--++(-4*\s-\p,-\p-4*\s)--++(\p+\p,-\p-\p)
                      --++(4*\s+2*\p,4*\s+2*\p)--++(-2*\p,2*\p)--++(-\s,-\s);
\draw[-,rounded corners=2mm,shorten >=2pt]
    ($(A10)+(-\p,\p)$)--++(-4*\s-\p,-\p-4*\s)--++(\p+\p,-\p-\p)--++(4*\s+2*\p,4*\s+2*\p)--++(-2*\p,2*\p)--++(-\s,-\s);
\draw[-,rounded corners=2mm,shorten >=2pt]
    ($(A40)+(-\p,\p)$)--++(-4*\s-\p,-\p-4*\s)--++(\p+\p,-\p-\p)--++(4*\s+2*\p,4*\s+2*\p)--++(-2*\p,2*\p)--++(-\s,-\s);
\end{tikzpicture}
\caption{Part of the graph $\DD{\G}{\e}$ presented in Fig.~\ref{G_konstr} that
illustrates the $\mathcal{H}^{(\ast)}$-partition $\picut{\pi(\F)}{\gadget_v}$, which gives the $\LL$-orientation of $\gadget_v$
in $\pi(\F)$ on the left hand side, and a example that illustrates a not well-oriented gadget $\gadget$ in the example
partition $\mathcal{H}^{(\ast)}$-partition of $\gadget$ on the right hand side.}
\label{Orien}
\end{figure}}
\renewcommand{\u}{u}
\renewcommand{\r}{\mathfrak{r}}
\newcommand{\D}{{\DD{\G}{\e}}}

\begin{lem}\label{modyf}
Let $\pi$ be an $\mathcal{H}^{(\ast)}$--partition of $\DD{\G}{\e}$ for which the gadget $\gadget_u$ is not
well-oriented where $u\in \V$. Then there exists an $\mathcal{H} ^{(\ast)}$--partition $\pi^\prime$ of $\DD{\G}{\e}$
which satisfies $|\pi^\prime|\leq |\pi|$, $\gadget_u$ is well-oriented in $\pi^\prime$, and also preserves the orientation of
$\gadget_w$ for $w \neq u,\, w\in \V$ ({i.e.\ $\gadget_w$ is well-oriented in $\pi$ if and only if $\gadget_w$ is well-oriented
in $\pi^\prime$}).
\end{lem}
\begin{proof}
Let $\pi$, $\u$ fulfil the assumptions. Our proof starts with the observation that vertices of $\DD{\G}{\e}_{|\pi,\u}$ can be
divided into three types: elements of $\gadgetV_\u$, elements which can be represented in the form $\langle \u,l\rangle$ for
some $l\in N^+_\G(\u)$, and elements which can be represented in the form $\langle r,\u\rangle$ for some $r\in N^-_\G(\u)$. Let
$L\coloneq\{\langle \u,l\rangle:\langle \u,l\rangle\in \VV_{|\pi,u}\}$, $R\coloneq\{\langle r,\u\rangle:\langle r,\u\rangle\in
\VV_{|\pi,\u}\}$ and suppose that $L=\{\langle \u,l_1\rangle, \langle \u,l_2\rangle, \ldots,\langle \u,l_i\rangle\}$,
$R=\{\langle r_1,\u\rangle, \langle r_2,\u\rangle,\ldots,\langle r_j,\u\rangle\}$, where $\ee{\langle
\u,l_1\rangle}<\ee{\langle \u,l_2\rangle}<\ldots<\ee{\langle \u,l_i\rangle}$, {and} $\ee{\langle r_1,\u\rangle}<\ee{\langle
r_2,\u\rangle}<\ldots<\ee{\langle r_j,\u\rangle}$. Since $\G$ is without self-loops, we infer that  $\{l_1,l_2,\ldots, l_i\}$
$\cap \{r_1,r_2,\ldots,r_j\}=\emptyset$. We also denote by $\pipathcut{\pi}{\u}{P}$ the subsequence of $\pipath{\pi}{P}$ which
is an $\EE$-path given by $\vertexOF{\pipath{\pi}{P}}\cap \gadget_\u$, since $\vertexOF{\pipath{\pi}{P}}\setminus \gadget_\u$
can contain at most sources in $\D$ if $\vertexOF{\pipath{\pi}{P}}\cap \gadget_\u\neq \emptyset$, where  $P\in \pi$.
Additionally, we denote by $\pipathcut{\pi}{\u}{v}$ the path $\pipathcut{\pi}{\u}{P}$, where $v\in P$.

\par
{\em Step 1.} We show that $|\picut{\pi}{\u}|\geq |\E|+1+\min\{i,j\}$. To do this, suppose first that for every integer $k$
that satisfies $1\leq k\leq |\E|$ there exists $P\in \picut{\pi}{\u}$ for which $\pipathcut{\pi}{\u}{P}$ is a non-empty
subsequence of $\LL_\u^{k}$. Thus the partition $\pi$ has at least $|\E|$ elements, i.e.\ the ones forming the subsequence of
 $\LL_\u^{k}$ for $1\leq k\leq |\E|$. {These elements} are also
different from any $\vertexOF{\pipath{\pi}{v}}$ where $v \in R\cup \{\u_{0,0}\}$, since $\vertexOF{\LL_\u^{k}}\cap (R\cup
\{\u_{0,0}\}) = \emptyset$ for each $k$. {In addition,} $\vertexOF{\pipath{\pi}{v_1}}\neq \vertexOF{\pipath{\pi}{v_2}}$ for
different $v_1,v_2\in R\cup \{\u_{0,0}\}$, as $v_1$, $v_2$ have no incoming arcs in $\DD{\G}{\e}_{|\pi,\u}$.
 Thus finally $|\picut{\pi}{\u}|\geq |\E|+|R\cup \{u_{0,0}\}|=|\E|+1+j\geq |\E|+1+\min\{i,j\}$.

Suppose now that there exists $k$ that satisfies $1\leq k\leq |\E|$, for which there is no $P\in\picut{\pi}{\u}$ such that the
whole sequence $\pipathcut{\pi}{\u}{P}$ is a subsequence of $\LL^k_\u$. From Lemma~\ref{R2_S1_WN1} we have that for each $1\leq
d \leq |\E|$, $\pipath{\pi}{\u_{k,d}}$ is a subsequence of $\LL^k_\u$ or $\RR^d_u$, hence by assumption
$\vertexOF{\pipath{\pi}{\u_{k,d}}}\subseteq
\vertexOF{\RR^d_u}$ for each $d$.
Thus $\pipath{\pi}{\u_{k,1}}, \pipath{\pi}{\u_{k,2}},\ldots,\pipath{\pi}{\u_{k,|\E|}}$ are pairwise different and
$\pipath{\pi}{\u_{k,d}}\neq \pipath{\pi}{v}$ for every $1\leq d\leq |\E|,\,v\in L\cup \{\u_{0,0}\}$ (note that the possibility
$\pipath{\pi}{\u_{k,0}}=\pipath{\pi}{v}$ for some $v\in L\cup \{\u_{0,0}\}$ is not excluded). Additionally,
$\pipath{\pi}{v_1}\neq \pipath{\pi}{v_2}$ for different $v_1,v_2\in L\cup \{\u_{0,0}\}$ as $v_1$, $v_2$ have no incoming arcs,
hence $|\picut{\pi}{\u}|\geq |\E|+|L\cup \{\u_{0,0}\}|= |\E|+1+i\geq |\E|+1+\min\{i,j\}$ and the proof of this step is
complete.

\par
{\em Step 2.} {We describe here a property of arcs in $\pidigraph{\D}{\pi}$ whose both tail and head cross different gadgets.
Let $\langle P_1,P_2\rangle$ be an arc of $\pidigraph{\D}{\pi}$ such that $P_1\cap \gadgetV_{v}\neq \emptyset$, $P_2\cap
\gadgetV_{w}\neq \emptyset$, and $v\neq w$, where $v,w\in\V$. We show that it is possible only in two cases, depending on the
direction of the arc between $v$ and $w$:
\begin{enumerate}[label=(\roman*)]
\item $\langle v,w\rangle,\:v_{\ee{\langle v,w\rangle},0}\in P_1$,
$w_{0,\ee{\langle v,w\rangle}}\in P_2$,
\item or $\langle w,v\rangle,\:v_{0,\ee{\langle w,v\rangle}}\in P_1$,
$w_{\ee{\langle w,v\rangle},0}\in P_2$.
\end{enumerate}
Observe that since $\gadgetV_{v}\pppp{\EE}\gadgetV_{w}=\emptyset$ and $P_1\pppp{\EE}P_2\neq\emptyset$, one of vertices in the
arc that connects $P_1$ with $P_2$ must be outside of $\gadgetV_{v}$, $\gadgetV_{w}$ (see Fig.~\ref{mody_step2}). This vertex
has to be a part of a Hamiltonian path intersecting the sets $\gadgetV_{v}$, $\gadgetV_{w}$, hence it has the form
$\langle{v},{x}\rangle$, $\langle{w},{x}\rangle$, $\langle{x},{v}\rangle$, or $\langle{x},{w}\rangle$. The only vertices of the
form that can be connected to both $\gadgetV_{v}$ and $\gadgetV_{w}$ are $\langle{v},{w}\rangle$ and $\langle{w},{v}\rangle$.
Therefore, we have
 $\langle v,w\rangle \in P_1\cup P_2$
or $\langle w,v\rangle \in P_1\cup P_2$. We show that $(i)$ holds for $\langle v,w\rangle \in P_1\cup P_2$ and $(ii)$ in the
second case. We prove only the first case, the second case is analogous. Let us assume that $\langle v,w\rangle \in P_1\cup
P_2$ and suppose first, contrary to our claim that $\langle v,w\rangle\notin P_1$. Then $\langle v,w\rangle\in P_2$. Obviously,
$\langle v,w\rangle$ is adjacent only with $v_{\ee{\langle v,w\rangle},0}$, $w_{0,\ee{\langle v,w\rangle}}$ and both cannot
belong at the same time to $P_2$, since $\langle v,w\rangle$ is a source in $\D$ and $P_2$ determines the path
$\pipath{\pi}{\langle v,w\rangle}$. Note, that an arbitrary directed path of $\D$ can go through at most one vertex outside of
the gadgets of the form \eqref{Gadget_def_5.2}. Additionally, there are no arcs that connect vertices belonging to two
different gadgets. Hence $P_2\setminus\{\langle v,w\rangle\}\subseteq \gadgetV_w$ since $\langle v,w\rangle\in P_2$ and
$P_2\cap \gadgetV_{w}\neq \emptyset$. Consequently, $\langle v,w\rangle$ is the only vertex that belongs to $P_2$ and is
adjacent to any vertex of $\gadgetV_v$. Then, $\langle\langle v,w\rangle,v_{\ee{\langle v,w\rangle},0}\rangle
\in P_2\pppp{\EE}P_1$ and finally,
$\langle P_2,P_1\rangle$ is an arc of $\pidigraph{\D}{\pi}$, but this contradicts the assumption that $\pidigraph{\D}{\pi}$ is
acyclic. Since $\langle v,w\rangle\in P_1$, it is easy to check that $v_{\ee{\langle v,w\rangle},0}\in P_1$, $w_{0,\ee{\langle
v,w\rangle}}\in P_2$. }
\begin{figure}
\begin{center}
\begin{tikzpicture}
\def\s{0.4cm}
\def\p{0.125cm}
\def\q{0.135cm}
\tiny
[>=stealth,->,shorten >=2pt,looseness=.5,auto]
    \matrix (M)[matrix of math nodes,
        column sep={\s,between origins},
        row sep={\s,between origins},
        nodes={mynodes,color=black}]
{ &&&&            &&|(C00)| &&   &&&&&  \\ &&&&            &&& &   &&&&&  \\ &&|(A00)|&&       &&&&   &&|(B00)|&&&         \\
&&&|(A10)|&       &&&&   &|(B01)|&&&         \\ |(A02)|&&&&|(A20)| &&&&   |(B02)|&&&&|(B20)|&  \\ &&&&              &&&&
&&&&&   \\ &&|(A22)|&&       &&&&   &&|(B22)|&&&  \\ };
\draw[-,rounded corners=2mm,color=black, shorten >=2pt]
    ($(A00)+(-\q,\q)$)--++
    (-2*\s-\q,-\q-2*\s)--++
    (\q+\q+2*\s,-\q-\q-2*\s)--++
    (2*\s+2*\q,2*\s+2*\q)--++
    (-2*\q-2*\s,2*\q+2*\s)--++(-\s,-\s);
\draw[-,rounded corners=2mm,color=black, shorten >=2pt]
    ($(B00)+(-\q,\q)$)--++
    (-2*\s-\q,-\q-2*\s)--++
    (\q+\q+2*\s,-\q-\q-2*\s)--++
    (2*\s+2*\q,2*\s+2*\q)--++
    (-2*\q-2*\s,2*\q+2*\s)--++(-\s,-\s);
\draw[-,rounded corners=2mm, shorten >=0pt]
    ($(A10)+(-2*\p,0)$)--++
    (3*\s+2*\p,2*\p+3*\s)--++
    (\p+\p,-\p-\p)--++
    (-3*\s-2*\p,-2*\p-3*\s);
\draw[-,rounded corners=2mm, shorten >=0pt]
    ($(B01)+(0,-2*\p)$)--++
    (-2*\p,2*\p)--++
    (\p+\p,\p+\p)--++
    (2*\p,-2*\p);
\node at (C00) [above left,inner sep=1.5mm]{\small $\langle v, w \rangle$};
\node at ($(A10)+(0,0.18)$) [right,inner sep=3mm]{$v_{\ee{\langle v, w \rangle},0}$};
\node at ($(B01)+(0,-0.12)$) [left,inner sep=2mm]{$w_{0,\ee{\langle v, w \rangle}}$};
\node at ($(A10)+(0.1,0)$) [above,inner sep=5mm]{\small $P_1$};
\node at ($(B01)+(-0.1,0)$) [above,inner sep=3mm]{\small $P_2$};
\node at ($(A02)+(0.3,0)$) [above,inner sep=8.5mm,color=black]{\small $\gadgetV_{v}$};
\node at ($(B20)+(-0.2,0)$) [above,inner sep=8.5mm,color=black]{\small $\gadgetV_{w}$};
\draw[Referencja](C00)--(A10);
\draw[Referencja](C00)--(B01);
\draw[Referencja,dotted,color = black](A00) -- (A10);
\draw[Referencja,dotted,color = black](A00) -- ($(A00)+(-\s+0.4*\p,-\s+0.4*\p)$);
\draw[Referencja,dotted,color = black](A10) -- (A20);
\draw[Referencja,dotted,color = black](A10) -- ($(A10)+(-\s+0.4*\p,-\s+0.4*\p)$);
\draw[Referencja,dotted,color = black](A02) -- ($(A02)+(\s-0.4*\p,-\s+0.4*\p)$);
\draw[Referencja,dotted,color = black]($(A02)+(\s-0.4*\p,\s-0.4*\p)$)--(A02);
\draw[Referencja,dotted,color = black](A20) -- ($(A20)+(-\s+0.4*\p,-\s+0.4*\p)$);
\draw[Referencja,dotted,color = black]($(A22)+(-\s+0.4*\p,\s-0.4*\p)$)--(A22);
\draw[Referencja,dotted,color = black]($(A22)+(\s-0.4*\p,\s-0.4*\p)$)--(A22);
\draw[Referencja,dotted,color = black](B00) -- (B01);
\draw[Referencja,dotted,color = black](B00) -- ($(B00)+(\s-0.4*\p,-\s+0.4*\p)$);
\draw[Referencja,dotted,color = black](B01) -- (B02);
\draw[Referencja,dotted,color = black](B01) -- ($(B01)+(\s-0.4*\p,-\s+0.4*\p)$);
\draw[Referencja,dotted,color = black](B02) -- ($(B02)+(\s-0.4*\p,-\s+0.4*\p)$);
\draw[Referencja,dotted,color = black]
($(B20)+(-\s+0.4*\p,\s-0.4*\p)$)--(B20);
\draw[Referencja,dotted,color = black](B20) -- ($(B20)+(-\s+0.4*\p,-\s+0.4*\p)$);
\draw[Referencja,dotted,color = black]($(B22)+(-\s+0.4*\p,\s-0.4*\p)$)--(B22);
\draw[Referencja,dotted,color = black]($(B22)+(\s-0.4*\p,\s-0.4*\p)$)--(B22);
\node at ($(C00)+(0,-3.2)$) [above]{\small the case $\langle v,w\rangle\in P_1\cup P_2$};
\end{tikzpicture}$\mkern40mu$\begin{tikzpicture}
\def\s{0.4cm}
\def\p{0.125cm}
\def\q{0.135cm}
\tiny
[>=stealth,->,shorten >=2pt,looseness=.5,auto]
    \matrix (M)[matrix of math nodes,
        column sep={\s,between origins},
        row sep={\s,between origins},
        nodes={mynodes,color=black}]
{ &&&&            &&|(C00)| &&   &&&&&  \\ &&&&            &&& &   &&&&&  \\ &&|(A00)|&&       &&&&   &&|(B00)|&&&         \\
&&&|(A10)|&       &&&&   &|(B01)|&&&         \\ |(A02)|&&&&|(A20)| &&&&   |(B02)|&&&&|(B20)|&  \\ &&&&              &&&&
&&&&&   \\ &&|(A22)|&&       &&&&   &&|(B22)|&&&  \\ };
\draw[-,rounded corners=2mm,color=black, shorten >=2pt]
    ($(A00)+(-\q,\q)$)--++
    (-2*\s-\q,-\q-2*\s)--++
    (\q+\q+2*\s,-\q-\q-2*\s)--++
    (2*\s+2*\q,2*\s+2*\q)--++
    (-2*\q-2*\s,2*\q+2*\s)--++(-\s,-\s);
\draw[-,rounded corners=2mm,color=black, shorten >=2pt]
    ($(B00)+(-\q,\q)$)--++
    (-2*\s-\q,-\q-2*\s)--++
    (\q+\q+2*\s,-\q-\q-2*\s)--++
    (2*\s+2*\q,2*\s+2*\q)--++
    (-2*\q-2*\s,2*\q+2*\s)--++(-\s,-\s);
\draw[-,rounded corners=2mm, shorten >=0pt]
    ($(B01)+(2*\p,0)$)--++
    (-3*\s-2*\p,2*\p+3*\s)--++
    (-\p-\p,-\p-\p)--++
    (3*\s+2*\p,-2*\p-3*\s);
\draw[-,rounded corners=2mm, shorten >=0pt]
    ($(A10)+(0,-2*\p)$)--++
    (2*\p,2*\p)--++
    (-\p-\p,\p+\p)--++
    (-2*\p,-2*\p);
\node at (C00) [above left,inner sep=1.5mm]{\small $\langle w, v \rangle$};
\node at ($(A10)+(0,0.18)$) [right,inner sep=3mm]{$w_{\ee{\langle w, v \rangle},0}$};
\node at ($(B01)+(0,-0.12)$) [left,inner sep=2mm]{$v_{0,\ee{\langle w, v \rangle}}$};
\node at ($(A10)+(0.1,0)$) [above,inner sep=3mm]{\small $P_2$};
\node at ($(B01)+(-0.1,0)$) [above,inner sep=5mm]{\small $P_1$};
\node at ($(A02)+(0.3,0)$) [above,inner sep=8.5mm,color=black]{\small $\gadgetV_{w}$};
\node at ($(B20)+(-0.2,0)$) [above,inner sep=8.5mm,color=black]{\small $\gadgetV_{v}$};
\draw[Referencja](C00)--(A10);
\draw[Referencja](C00)--(B01);
\draw[Referencja,dotted,color = black](A00) -- (A10);
\draw[Referencja,dotted,color = black](A00) -- ($(A00)+(-\s+0.4*\p,-\s+0.4*\p)$);
\draw[Referencja,dotted,color = black](A10) -- (A20);
\draw[Referencja,dotted,color = black](A10) -- ($(A10)+(-\s+0.4*\p,-\s+0.4*\p)$);
\draw[Referencja,dotted,color = black](A02) -- ($(A02)+(\s-0.4*\p,-\s+0.4*\p)$);
\draw[Referencja,dotted,color = black]($(A02)+(\s-0.4*\p,\s-0.4*\p)$)--(A02);
\draw[Referencja,dotted,color = black](A20) -- ($(A20)+(-\s+0.4*\p,-\s+0.4*\p)$);
\draw[Referencja,dotted,color = black]($(A22)+(-\s+0.4*\p,\s-0.4*\p)$)--(A22);
\draw[Referencja,dotted,color = black]($(A22)+(\s-0.4*\p,\s-0.4*\p)$)--(A22);
\draw[Referencja,dotted,color = black](B00) -- (B01);
\draw[Referencja,dotted,color = black](B00) -- ($(B00)+(\s-0.4*\p,-\s+0.4*\p)$);
\draw[Referencja,dotted,color = black](B01) -- (B02);
\draw[Referencja,dotted,color = black](B01) -- ($(B01)+(\s-0.4*\p,-\s+0.4*\p)$);
\draw[Referencja,dotted,color = black](B02) -- ($(B02)+(\s-0.4*\p,-\s+0.4*\p)$);
\draw[Referencja,dotted,color = black]
($(B20)+(-\s+0.4*\p,\s-0.4*\p)$)--(B20);
\draw[Referencja,dotted,color = black](B20) -- ($(B20)+(-\s+0.4*\p,-\s+0.4*\p)$);
\draw[Referencja,dotted,color = black]($(B22)+(-\s+0.4*\p,\s-0.4*\p)$)--(B22);
\draw[Referencja,dotted,color = black]($(B22)+(\s-0.4*\p,\s-0.4*\p)$)--(B22);
\node at ($(C00)+(0,-3.2)$) [above]{\small the case $\langle w,v\rangle\in P_1\cup P_2$};
\end{tikzpicture}
\end{center}
\caption{An illustration of the construction from the proof of Lemma~\ref{modyf}, step 2.}
\label{mody_step2}
\end{figure}
\par
{\em Step 3.} {We prove here the final conclusion.} Let us consider two $\mathcal{H}^{(\ast)}$-partitions of
$\D_{|\pi,\u}$
\begin{equation}
\begin{array}{rcl}
\widehat{\mathcal{L}}_\u&\coloneq&\{\vertexOF{{\LL}_\u^{n}}:0\leq n\leq |\E|\,\wedge\,n\notin\ee{L}\}
\cup\\
&&
\strut\mkern25mu\{\vertexOF{\langle\u,l_n\rangle\rangle\!\mathbin{{}^\frown}\!\LL_\u^{\ee{\langle \u,l_n\rangle}}}:1\leq n\leq
i\}
\cup\{\{\langle r_n,\u\rangle\}:1\leq n\leq j\},\\
\widehat{\mathcal{R}}_\u&\coloneq&\{\vertexOF{\RR_\u^{n}}:0\leq n\leq |\E|\,\wedge\,n\notin\ee{R}\}\cup\\
&&
\strut\mkern25mu\{\vertexOF{\langle\langle r_n,\u\rangle\rangle\!\mathbin{{}^\frown}\!\RR_\u^{\ee{\langle
r_n,\u\rangle}}}:1\leq n\leq j\}
\cup\{\{\langle \u,l_n\rangle\}:1\leq n\leq i\}.\\
\end{array}
\end{equation}
It is clear that $|\widehat{\mathcal{L}}_\u|=|\E|\!+\!1\!+\!j$, $|\widehat{\mathcal{R}}_\u|=|\E|\!+\!1\!+\!i$. From the
assumption about $\G$ we obtain that in-degree or out-degree of $\u$ is $1$. {Let us}  assume that in-degree of $\u$ is $1$,
thus
 there exists $\r$ such that $N_\G^-(\u)=\{\r\}$ (the second case $|N_\G^+(\u)|=1$ is analogous).
The situation that $\langle\r,u\rangle \notin \VV_{|\pi,u}$
 is not excluded here, hence $j\leq 1$.
We split the proof into the following cases:
\begin{enumerate}[label=\arabic*.]

\item {\em Suppose that $j=1, i\geq 1$}. From {\em Step 1.}
we have that $|\picut{\pi}{\u}|\geq |\E|+2=|\widehat{\mathcal{L}}_\u|$, and consequently $|\pi^\prime|\leq |\pi|$ where
$\pi^\prime$ is defined by $\pi^\prime \coloneq (\pi\setminus\picut{\pi}{\u})\cup
\widehat{\mathcal{L}}_\u$. Obviously, $\pi^\prime$ is an
$\mathcal{H}^{\ast}$-partition of $\D$ and $\pi^\prime$ fulfils the condition about well-ordering of gadgets. To
complete the proof it is enough to show that $\pidigraph{\D}{\pi^\prime}$ is acyclic. Suppose, contrary to our claim, that
there exists a cycle $\mathfrak{a}$ in $\pidigraph{\D}{\pi^\prime}$. {As} $\pidigraph{\D}{\pi}$ is acyclic we have that
$\mathfrak{a}$ has to cross vertices of $\widehat{\mathcal{L}}_\u$. To obtain a contradiction, we show that it is possible
to replace every maximal subpath of $\mathfrak{a}$ that consists of vertices from $\widehat{\mathcal{L}}_\u$ by a path
which is obtained from vertices of $\picut{\pi}{\u}$ in such a way that the modified path $\mathfrak{a}^\prime$ is a cycle.
By repetitive application of the procedure we obtain a cycle in
 $\pidigraph{\D}{\pi}$.

\par
Let us take such a maximal subpath $\mathfrak{a}^\prime=\langle a^\prime_1,a^\prime_2,\ldots,a^\prime_k\rangle$ of
$\mathfrak{a}$ with elements of $\widehat{\mathcal{L}}_\u$ and denote by $a^\prime_0$, $a^\prime_{k+1}$ vertices that
belong to $\mathfrak{a}$ and $\langle a^\prime_0,a^\prime_1\rangle$, $\langle a^\prime_k,a^\prime_{k+1}\rangle$ are arcs of
$\pidigraph{\D}{\pi^\prime}$. We have that $\mathfrak{a}\neq \mathfrak{a}^\prime$ and $a^\prime_0,a^\prime_{k+1}\neq
a^\prime_1,a^\prime_2,\ldots,a^\prime_k$
 since $\widehat{\mathcal{L}}_\u$
is an $\mathcal{H} ^{(\ast)}$-partitions of $\D_{|\pi,\u}$. Since $\{\langle \r,\u\rangle\}$ is a source in
$\pidigraph{\D}{\pi^\prime}$ it cannot occur in a cycle, therefore for each $i=1,2,\ldots,k$
  the property $a^\prime_i\cap \gadgetV_\u\neq \emptyset$ holds. For similar reasons,
  the edges in $\E$ that witness
  $\langle a^\prime_0,a^\prime_1\rangle$, $\langle a^\prime_k,a^\prime_{k+1}\rangle$
  have the form
    $\langle \u,v\rangle,\langle \u,w\rangle\in \E$ where
$a^\prime_0 \cap \gadgetV_v\neq \emptyset$, $a^\prime_{k+1} \cap \gadgetV_w\neq \emptyset$. Then from {\em Step 2.} we have
that $\langle \u,v\rangle \in a^\prime_0$, $\u_{\ee{\langle \u,v\rangle},0}\in a_1^\prime$, $\langle
\u,w\rangle,\,\u_{\ee{\langle \u,w\rangle},0} \in a^\prime_k$, $w_{0,\ee{\langle \u,w\rangle}}\in a_{k+1}^\prime$ and, in
consequence,
 $v\in L$, $w\in \mathcal{N}^+_\G(\u)\setminus L$.
 Additionally,
$\ee{\langle \u,v\rangle} < \ee{\langle \u,w\rangle}$ (see Fig.~\ref{mody_step3_case1}). Let us consider a sequence
\begin{equation}
\mathfrak{b}=\langle\pipath{\pi}{\u_{\ee{\langle \u,v\rangle},0}},
\pipath{\pi}{\u_{\ee{\langle \u,v\rangle}+1,0}},
\pipath{\pi}{\u_{\ee{\langle \u,v\rangle}+2,0}},
\ldots,
\pipath{\pi}{\u_{\ee{\langle \u,w\rangle},0}}\rangle
\end{equation}
of vertices of $\picut{\pi}{\u}$. Obviously, vertices of this sequence do not have to be pairwise different. Therefore let
us consider the maximal subsequence $\mathfrak{b}^\prime=\langle b_1,b_2,\ldots, b_{k^\prime} \rangle$ of $\mathfrak{b}$
such that $b_i\neq b_{i+1}$ for $i=1,2,\ldots,k^\prime-1$. It is easy to check that for each $i$ there exists $j$ such that
$\ee{\langle \u,v\rangle}\leq j<{\ee{\langle \u,w\rangle}}$, $b_i=\pipath{\pi}{\u_{j,0}}$,
$b_{i+1}=\pipath{\pi}{\u_{j+1,0}}$, hence $\langle b_i, b_{i+1}\rangle$ is an arc of $\pidigraph{\D}{\pi}$, since
$\langle\u_{j,0},\u_{j+1,0}\rangle\in \EE$, and finally $\mathfrak{b}^\prime$ is a path of $\pidigraph{\D}{\pi}$.
Additionally, $\langle a_0^\prime,b_1\rangle$, $\langle b_{k^\prime},a^\prime_{k+1}\rangle$ are arcs of
$\pidigraph{\D}{\pi}$, hence we obtain that the subpath $\mathfrak{a}^\prime$ can be replaced by $\mathfrak{b}^\prime$.

\begin{figure}
\begin{center}
$\strut\mkern190mu$\begin{tikzpicture}
\def\s{0.43cm}
\def\p{0.125cm}
\def\q{0.135cm}
\tiny
[>=stealth,->,shorten >=2pt,looseness=.5,auto]
    \matrix (M)[matrix of math nodes,
        column sep={\s,between origins},
        row sep={\s,between origins},
        nodes={mynodes,color=black}]
{ &&&&&|(uv)| && && &&&&&&\\ &       &&&&&|(C00)|& && &&&&\\ &|(A00)|&&&&&&& && \\ &&|(A10)|&&& && &\\ &&&|(A20)|&& &&
&&|(uw)|\\ &&&&& && &&&|(B00)|\\ &&       &&&|(A30)| && && \\ &&&&&&|(A40)| && && \\ &&&&&&&|(A50)|& &&\\ };
\draw[Referencja,dotted,color = black](A00) -- (A10);

\draw[Referencja,color = black](A10) -- (A20);
\draw[Referencja,dotted,color = black](A20) -- (A30);
\draw[Referencja,color = black](A30) -- (A40);
\draw[Referencja,dotted,color = black](A40) -- (A50);
\draw[Referencja](uv) -- (A10);
\draw[Referencja](uv) -- (C00);
\draw[Referencja](uw) -- (A40);
\draw[Referencja](uw) -- (B00);
\draw[Referencja,dotted,color = black]
(A10) -- ($(A10)+(-\s+0.4*\p,-\s+0.4*\p)$);
\draw[Referencja,dotted,color = black]
(A20) -- ($(A20)+(-\s+0.4*\p,-\s+0.4*\p)$);
\draw[Referencja,dotted,color = black]
(A30) -- ($(A30)+(-\s+0.4*\p,-\s+0.4*\p)$);
\draw[Referencja,dotted,color = black]
(A40) -- ($(A40)+(-\s+0.4*\p,-\s+0.4*\p)$);
\draw[Referencja,dotted,color = black]
(A50) -- ($(A50)+(-\s+0.4*\p,-\s+0.4*\p)$);
\draw[Referencja,dotted,color = black]
(A00) -- ($(A00)+(-\s+0.4*\p,-\s+0.4*\p)$);

\draw[-,rounded corners=2mm,color=black, shorten >=0pt]
    ($(A00)+(-\q-\s,\q-\s)$)--++
    (\q+\s,\q+\s)--++
    (\q+6*\s+\q,-\q-6*\s-\q)--++
    (-\q-\s,-\q-\s);
\draw[-,rounded corners=2mm,color=black, shorten >=0pt]
    ($(B00)+(-\s*0.5,-2*\q-\s*0.5)$)--++
    (-2*\q,2*\q)--++
    (2*\q+\s,2*\q+\s)--++
    (2*\q,-2*\q);
 \draw[-,rounded corners=2mm,color=black, shorten >=0pt]
    ($(C00)+(-\s*0.5,-2*\q-\s*0.5)$)--++
    (-2*\q,2*\q)--++
    (2*\q+\s,2*\q+\s)--++
    (2*\q,-2*\q);
\draw[-,rounded corners=2mm, shorten >=0pt]
    ($(A40)+(-2*\p,0)$)--++
    (3*\s+2*\p,2*\p+3*\s)--++
    (\p+\p,-\p-\p)--++
    (-3*\s-2*\p,-2*\p-3*\s);
\draw[-,rounded corners=2mm, shorten >=0pt]
    ($(C00)+(2*\p,0)$)--++
    (-1*\s-2*\p,2*\p+1*\s)--++
    (-\p-\p,-\p-\p)--++
    (1*\s+2*\p,-2*\p-1*\s);
\draw[-,rounded corners=2mm, shorten >=0pt]
    ($(B00)+(2*\p,0)$)--++
    (-2*\p,2*\p)--++
    (-\p-\p,-\p-\p)--++
    (2*\p,-2*\p);
    \draw[-,rounded corners=2mm, shorten >=0pt]
    ($(A30)+(-2*\p,0)$)--++
    (2*\p,2*\p)--++
    (\p+\p,-\p-\p)--++
    (-2*\p,-2*\p);
    \draw[-,rounded corners=2mm, shorten >=0pt]
    ($(A20)+(-2*\p,0)$)--++
    (2*\p,2*\p)--++
    (\p+\p,-\p-\p)--++
    (-2*\p,-2*\p);
\draw[-,rounded corners=2mm, shorten >=0pt]
    ($(A10)+(-2*\p,0)$)--++
    (2*\p,2*\p)--++
    (\p+\p,-\p-\p)--++
    (-2*\p,-2*\p);
\node at (uv) [color=black,above,inner sep=2.2mm]{ $\langle u, v \rangle$};
\node at (uw) [color=black,above,inner sep=2.2mm]{ $\langle u, w \rangle$};

\node at ($(C00)+(0.05,0.1)$) [above,inner sep=3.2mm]{\small $a_0^\prime$};
\node at ($(B00)+(0.05,-0.05)$) [above right,inner sep=1.2mm]{\small $a_{k+1}^\prime$};

\node at ($(A10)+(0,-0.1)$) [color=black,left,inner sep=1.9mm]
{ $u_{\ee{\langle u, v \rangle},0}$};
\node at ($(A20)+(0,-0.1)$) [color=black,left,inner sep=2.9mm]
{$u_{\ee{\langle u, v \rangle}\!+\!1,0}$};
\node at ($(A30)+(0,-0.1)$) [left,color=black,inner sep=2.9mm]
{$u_{\ee{\langle u, w \rangle}\!-\!1,0}$};
\node at ($(A40)+(0,-0.1)$) [left,inner sep=2.9mm,color=black]
{ $u_{\ee{\langle u, w \rangle},0}$};
\node at ($(A00)+(0,0.02)$) [left,inner sep=2.9mm,color=black]
{$u_{0,0}$};
\node at ($(A50)+(0,-0.1)$) [left,inner sep=2.9mm,color=black]
{$u_{|\E|+1,0}$};
\node at ($(C00)+(0,-0.2)$) [right,inner sep=1.9mm,color=black]
{$v_{0,\ee{\langle u, v \rangle}}$};
\node at ($(B00)+(0,-0.2)$) [right,inner sep=1.9mm,color=black]
{$w_{0,\ee{\langle u, w \rangle}}$};
\node at ($(A10)+(0,0.1)$) [right,inner sep=2.5mm]
{\small $\pipath{\pi}{\u_{\ee{\langle \u,v\rangle},0}}$};
\node at ($(A20)+(0,0.1)$) [right,inner sep=1.7mm]
{\small $\pipath{\pi}{\u_{\ee{\langle \u,v\rangle}\!+\!1,0}}$}; 
\node at ($(A40)+(0,-0.1)$) [right,inner sep=3.4mm]
{\small $\pipath{\pi}{\u_{\ee{\langle \u,w\rangle}\!+\!1,0}}$};
\node at ($(A00)+(0,0.2)$) [color=black, above left,inner sep=1.4mm]
{\small $\gadgetV_u$};
\node at ($(uv)+(0.7,-0.1)$) [color=black, above right,inner sep=1.4mm]
{\small $\gadgetV_v$};
\node at ($(uw)+(0.7,-0.1)$) [color=black, above right,inner sep=1.4mm]
{\small $\gadgetV_w$};
\end{tikzpicture}\end{center}
\caption{An illustration of the construction from the proof of Lemma~\ref{modyf}, step 3, case 1.}
\label{mody_step3_case1}
\end{figure}

\item {\em Suppose that $j=i=0$}.
From {\em Step 1.} we have that $|\picut{\pi}{\u}|\geq |\E|+1=|\widehat{\mathcal{L}}_\u|$. Then the proof completes if we
take $\pi^\prime \coloneq (\pi\setminus\picut{\pi}{\u})\cup \widehat{\mathcal{L}}_\u$, since every arc of
$\pidigraph{\D}{\pi^\prime}$ that joins vertices of $\pi\setminus\picut{\pi}{\u}$ and $\widehat{\mathcal{L}}_\u$ is
directed towards $\widehat{\mathcal{L}}_\u$.

\item {\em Suppose that $j=0$, $i\geq1$.}
From {\em Step 1.} we have that $|\picut{\pi}{\u}|\geq |\E|+1=|\widehat{\mathcal{L}}_\u|$. Let $\pi^\prime \coloneq
(\pi\setminus\picut{\pi}{\u})\cup \widehat{\mathcal{L}}_\u$. Obviously, $|\pi^\prime|\leq |\pi|$ and $\pi^\prime$ is an
$\mathcal{H}^{\ast}$-partition of $\D$ that fulfils the condition about well-ordering of gadgets. It is left to show
that $\pidigraph{\D}{\pi^\prime}$ is acyclic. Suppose, contrary to our claim, that there exists a cycle $\mathfrak{a}$ of
$\pidigraph{\D}{\pi^\prime}$. Similar to the first case, we infer that $\mathfrak{a}$ has to cross vertices of
$\widehat{\mathcal{L}}_\u$, but since $j=0$, $\{\langle
\r,\u\rangle\}$ need not be a source in
$\pidigraph{\D}{\pi^\prime}$, and $\mathfrak{a}$ can cross $\widehat{\mathcal{L}}_\u$, not only by going in and going out
from the ``right hand side'' of $\gadget_u$ (as in Fig.~\ref{mody_step3_case1}), but also by going in from the ``left hand
side'' of $\gadget_u$ traversing the vertex $P=\vertexOF{\pipath{\pi^\prime}{\langle
\r,\u\rangle}}$. In the case, when there exists a cycle of
$\pidigraph{\D}{\pi^\prime}$ that goes to $\widehat{\mathcal{L}}_\u$ traversing $P$, we prove that $|\pi^\prime|<|\pi|$.
Then we conclude that $\pi^{\prime\prime}=(\pi^\prime\setminus\{P\})\cup\{P\setminus\{\langle
\r,\u\rangle\}, \{\langle \r,\u\rangle\}\}$, which is an
\mbox{$\mathcal{H}^{\ast}$-partition} of $\D$ that is
created by dividing $P$ into two parts $P\setminus\{\langle
\r,\u\rangle\}$, $\{\langle \r,\u\rangle\}$, has at most $|\pi|$
elements and preserves the orientanion of gadgets. Additionally, $\{\langle \r,\u\rangle\}$ is a source in
$\pidigraph{\D}{\pi^{\prime\prime}}$. Now, every cycle has to cross vertices of $\widehat{\mathcal{L}}_\u$, go in and go
out from the ``right hand side'' of $\gadget_u$. The application of the proof idea of the first case completes the proof.

Let us take a cycle of $\pidigraph{\D}{\pi^\prime}$ that goes into $\widehat{\mathcal{L}}_\u$ traversing $P$. Now we prove
that $|\pi^{\prime}| < |\pi|$. For this purpose it is enough to show that
$|\picut{\pi}{\u}|>|\E|+1=|\widehat{\mathcal{L}}_\u|$. Suppose, contrary to our claim, that $|\picut{\pi}{\u}|\leq|\E|+1$.
From Lemma~\ref{R2_S1_T1} we conclude that $|\E|+1\leq |\picut{\pi}{\gadget_\u}|\leq |\picut{\pi}{\u}|\leq |\E|+1$ and in
addition $\picut{\pi}{\gadget_\u}=\mathcal{L}_\u$ or $\picut{\pi}{\gadget_\u}=\mathcal{R}_\u$. Note in this case also
$|\picut{\pi}{\gadget_\u}|=|\picut{\pi}{\u}|$. Thus, when $\picut{\pi}{\gadget_\u}=\mathcal{L}_\u$ we get equalities
 $\picut{\pi}{\u}=\widehat{\mathcal{L}}_\u$,
$\pi^\prime=\pi$, but this contradicts our assumption that $\pidigraph{\D}{\pi}$ is acyclic. In consequence,
$\picut{\pi}{\gadget_\u}=\mathcal{R}_\u$. Let us focus on the vertex $\langle u,l_1\rangle\in\VV_{|\pi,\u}$. This vertex is
adjacent only with $\u_{\ee{\langle \u,l_1\rangle},0}$ from $\VV_{|\pi,\u}$, which also belongs to $\vertexOF{\RR^0_\u}$.
Hence there exists $Q\in \picut{\pi}{\u}$  that contains
 $\{\langle u,l_1\rangle\}\cup \vertexOF{\RR^0_\u}$, since
 $|\picut{\pi}{\u}|=|\mathcal{R}_\u|$.
However $\{\langle u,l_1\rangle\}\cup \vertexOF{\RR^0_\u}$ contains two sources in $\D$, namely $\langle u,l_1\rangle$ and
$u_{0,0}$. This is in contradiction with the fact that all the vertices of $Q$ are in a single Hamiltonian path.

\item {\em Suppose that $j=1, i=0$}.
From {\em Step 1.} we have that $|\picut{\pi}{\u}|\geq |\E|+1=|\widehat{\mathcal{R}}_\u|$. Let $\pi^\prime \coloneq
(\pi\setminus\picut{\pi}{\u})\cup \widehat{\mathcal{R}}_\u$. Obviously, $|\pi^\prime|\leq |\pi|$ and $\pi^\prime$ is an
$\mathcal{H}^{\ast}$-partition of $\D$ that fulfils the condition about well-ordering of gadgets. We show that
$\pidigraph{\D}{\pi^\prime}$ is acyclic. Suppose, contrary to our claim, that there exists a cycle $\mathfrak{a}$ of
$\pidigraph{\D}{\pi^\prime}$. It is easy to check that $\mathfrak{a}$ has to go into $\widehat{\mathcal{L}}_\u$ from the
``left hand side'' of $\gadget_\u$ and go out through the vertex $P=\vertexOF{\pipath{\pi^\prime}{\langle \r,\u\rangle}}$
-- on the ``right hand side''. In such a case we show, similarly to the third case that $|\picut{\pi}{\u}|>|\E|+1$ and then
we conclude that $\pi^{\prime\prime}=(\pi^\prime\setminus\{P\})\cup\{P\setminus\{\langle
\r,\u\rangle\}, \{\langle \r,\u\rangle\}\}$ is an
\mbox{$\mathcal{H}^{\ast}$-partition} of $\D$ that has at
most $|\pi|$ elements, preserve the orientation of gadgets, and $\{\langle \r,\u\rangle\}$ is a source in
$\pidigraph{\D}{\pi^{\prime\prime}}$. Since we can modify $\pi$ in such a way that every cycle cannot go out from
$\widehat{\mathcal{R}}_\u$, the proof is completed.

Let us take a cycle of $\pidigraph{\D}{\pi^\prime}$ that goes out from $\widehat{\mathcal{R}}_\u$ through $P$, and suppose,
contrary to our claim, that $|\picut{\pi}{\u}|\leq|\E|+1$. From Lemma~\ref{R2_S1_T1} we conclude that $|\E|+1\leq
|\picut{\pi}{\gadget_\u}|\leq |\picut{\pi}{\u}|\leq |\E|+1$ and in addition $\picut{\pi}{\gadget_\u}=\mathcal{L}_\u$ or
$\picut{\pi}{\gadget_\u}=\mathcal{R}_\u$. Note that since $|\picut{\pi}{\gadget_\u}|=|\picut{\pi}{\u}|$ in the case
$\picut{\pi}{\gadget_\u}=\mathcal{R}_\u$ we get equalities
 $\picut{\pi}{\u}=\widehat{\mathcal{R}}_\u$,
$\pi^\prime=\pi$, but this contradicts our assumption that $\pidigraph{\D}{\pi}$ is acyclic. In consequence,
$\picut{\pi}{\gadget_\u}=\mathcal{L}_\u$. Let us focus on the vertex $\langle\r,\u\rangle\in\VV_{|\pi,\u}$. This vertex is
adjacent only with $\u_{0,\ee{\langle\r,\u\rangle}}$ from $\VV_{|\pi,\u}$, which also belongs to $\vertexOF{\LL^0_\u}$.
Hence there exists $Q\in \picut{\pi}{\u}$  that contains
 $\{\langle\r,\u\rangle\}\cup \vertexOF{\LL^0_\u}$, since
 $|\picut{\pi}{\u}|=|\mathcal{L}_\u|$.
However $\{\langle \r, \u\rangle\}\cup \vertexOF{\LL^0_\u}$ contains two sources in $\D$, namely $\langle \r,\u\rangle$ and
$\u_{0,0}$. This is in contradiction with the fact that all the vertices of $Q$ are in a single Hamiltonian path.

\end{enumerate}
\end{proof}

\begin{lem}\label{construction}
Let $\pi$ be an $\mathcal{H}^{(\ast)}$--partition of $\DD{\G}{\e}$ for which all gadgets are well-oriented. Then
$\{\langle v,u \rangle: \{\langle v,u\rangle\} \in \pi\}$ is a feedback arc set which has size at most
$|\pi|-|\V|\cdot(|\E|+1)$.
\end{lem}
\begin{proof}
Let $\pi$ fulfil the assumptions and introduce $\F\coloneq\{\langle v,u \rangle: \{\langle v,u\rangle\} \in \pi\}$,
$\mathcal{O}\coloneq\{\{\langle v,u\rangle\} \in \pi\}$. Suppose, contrary to our claim, that $\F$ is not a feedback arc set of
$\G$. Then there exists $\E\setminus\F$--path $\mathfrak{c}\coloneq\langle c_1,c_2,\ldots,c_k\rangle$ of $\G$ such that $k>2$
and $c_1=c_k$. Additionally, we let $c_{k+1}\coloneq c_2$ and we use the following notation:
\begin{equation}
\begin{array}{lcr}
\LL^{\langle v,u\rangle} & \coloneq &
\langle\langle v,u\rangle\rangle\!\mathbin{{}^\frown}\!\LL_v^{\ee{\langle v,u\rangle}},\\
^{\langle v,u\rangle}\RR & \coloneq &
\langle\langle v,u\rangle\rangle\!\mathbin{{}^\frown}\!\RR_u^{\ee{\langle v,u\rangle}},\\
\end{array}
\end{equation}
where $\langle v,u\rangle\in \E$. It is easy to check that for every $i$ which fulfils $1\leq i<k$, exactly one of two
conditions $\pipath{\pi}{\langle c_i,c_{i+1}\rangle}=
\LL^{\langle c_i,c_{i+1}\rangle}$,
$\pipath{\pi}{\langle c_i,c_{i+1}\rangle}=^{\langle c_i,c_{i+1}\rangle}\RR$ is satisfied, since all gadgets are well-oriented.
Moreover from the fact that no gadget has both orientations in $\pi$, we conclude that $\pipath{\pi}{\langle
c_i,c_{i+1}\rangle}= ^{\langle c_i,c_{i+1}\rangle}\RR$ if and only if $\pipath{\pi}{\langle c_{i+1},c_{i+2}\rangle}\neq
\LL^{\langle c_{i+1},c_{i+2}\rangle}$, for $i=1,2,\ldots,k-1$.
In consequence, every gadget $\gadget_{c_i}$ is $\LL$-oriented or every gadget $\gadget_{c_i}$ is $\RR$-oriented, therefore one
of the following sequences
\begin{multline}
\langle\LL_{c_1}^{0},\pipath{\pi}{(c_1)_{1,0}},
\ldots,
\pipath{\pi}{(c_1)_{\ee{\langle c_1,c_2\rangle}-1,0}},
\LL^{\langle c_1,c_2\rangle},\\
\LL_{c_2}^{0},\pipath{\pi}{(c_2)_{1,0}},
\ldots,
\pipath{\pi}{(c_1)_{\ee{\langle c_2,c_3\rangle}-1,0}},
\LL^{\langle c_2,c_3\rangle},\\
\ldots,
\LL_{c_{k-1}}^{0},\pipath{\pi}{(c_{k-1})_{1,0}},
\ldots,
\pipath{\pi}{(c_{k-1})_{\ee{\langle c_{k-1},c_k\rangle}-1,0}},
\LL^{\langle c_{k-1},c_k\rangle},\,\LL_{c_1}^{0}\rangle,
\end{multline}
\begin{multline}
\langle\RR_{c_k}^{0},\pipath{\pi}{(c_k)_{0,1}},
\ldots,
\pipath{\pi}{(c_k)_{0,\ee{\langle c_{k-1},c_k\rangle}-1}},
^{\langle c_{k-1},c_k\rangle}\RR,\\
\RR_{c_{k-1}}^{0},\pipath{\pi}{(c_{k-1})_{0,1}},
\ldots,
\pipath{\pi}{(c_{k-1})_{0,\ee{\langle c_{k-2},c_{k-1}\rangle}-1}},
^{\langle c_{k-2},c_{k-1}\rangle}\RR,\\
\ldots,
\RR_{c_2}^{0},\pipath{\pi}{(c_2)_{0,1}},
\ldots,
\pipath{\pi}{(c_2)_{0,\ee{\langle c_{1},c_2\rangle}-1}},
^{\langle c_{1},c_2\rangle}\RR,\,\RR_{c_k}^{0}\rangle
\end{multline}
represents a cycle of the \DAG{} $\pidigraph{\DD{\G}{\e}}{\pi}$, but $\pidigraph{\DD{\G}{\e}}{\pi}$ is acyclic since $\pi$
is\break an $\mathcal{H}^{(\ast)}$--partition of $\DD{\G}{\e}$, a contradiction. This finishes the proof, that $\F$ is a a
feedback arc set of $\G$.

\par It remains to show that
$\F\leq|\pi|-|\V|\cdot(|\E|+1)$, but it is obvious since {$\pi\setminus\mathcal{O}=\underset{v \in
\V}{\bigcup}\picut{\pi}{v}$;} $|\picut{\pi}{v}|=|\E|+1$, for every $v\in \V$; {and $\picut{\pi}{v_1}$, $\picut{\pi}{v_2}$ are
disjoint for every different $v_1,v_2\in \V$.}
\end{proof}

\begin{thm}\label{APHgeneral}
The \FAS{} problem is  reducible to the \AHP{} problem.
\end{thm}
\begin{proof}
For each instance $\mathcal{G}=\langle \V,\E\rangle$, $k$ of \FAS{} we can construct $\DD{\G}{\e}$, a number
$k+|\V|\cdot(|\E|+1)$ and consider them to be an instance of \AHP{}. This translation can clearly be done in logspace. By
Lemma~\ref{R2_S1_prosta_implikacja} if $\mathcal{G}$ has a feedback arc set of size not greater than $k$, then $\DD{\G}{\e}$
has a $\mathcal{H}^{(\ast)}$ partition of size not greater than $k+|\V|\cdot(|\E|+1)$. By Lemma~\ref{modyf} composed with
Lemma~\ref{construction} we obtain that when $\DD{\G}{\e}$ has a $\mathcal{H}^{(\ast)}$ partition of size not greater than
$k+|\V|\cdot(|\E|+1)$ then $\mathcal{G}$ has a feedback arc set of size not greater than $k$. Therefore translation
$\mathcal{G}$, $k$ to $\DD{\G}{\e}$, $k+|\V|\cdot(|\E|+1)$ is indeed a reduction.
\end{proof}

{According to the above theorem, \AHP{} is NP-complete. Additionally, vertices of the constructed digraph $\DD{\G}{\e}$ have
restricted in-degree and out-degree by 2. Hence we finally have the following theorem. }

\begin{thm}\label{APHlimit2}
The \AHP{} problem is NP-complete for digraphs in which no vertex has in-degree or out-degree greater than $2$.
\end{thm}

\section{Constructive abstract proof graphs}\label{Constructive_abstract_proof_graphs}

In Section \ref{Graph_representation_of_proofs} a special kind of \DAG{} was proposed to be a simplified model of natural
deduction proofs. The considered method describe{s} in a simple way a method of construction of the abstract proof graph for
every reasoning that does not contain nested subreasonings. This demonstrates that the model is sound. It is natural now to
give an argument that it is also complete, which means in this case that for every abstract proof graph there exists a
reasoning whose structure is that graph. Now we show that this question has a positive answer.

For our purposes we consider a constructive subfamily of \DAG{}s that includes graphs considered in Section~\ref{NPproof}. To
study the general case of this subfamily we have to assume that
 $E\neq \RRR{G}$ where $G=\langle V,E\rangle$
is an abstract proof graph
  with a distinguished set of arcs $\RRR{G}$ (see~p.~\pageref{RRRRR}) and
$\RRR{G}$ correspond to $E_1$ occurring in the instances of MIL problems. Since $\RRR{G}\subsetneq E$,
 the constructed reasoning has to have some arcs other than reference arcs.

These non-reference arcs are represented in the reasoning with the help of a special kind of steps in the Mizar system. This
special kind of steps in the general case has the shape:
\begin{equation}\label{step_private}
\begin{tabular}{c}
\miz{consider} {\em variable} \miz{be} {\em type} \miz{such} \miz{that}\\
$\strut\mkern20mu${\em label}\miz{:} {\em statement} \miz{by} {\em justification} \miz{;}\\
\end{tabular}
\end{equation}
These ``consider'' steps introduce new constants. Each step that uses such a constant is a head of a non-reference arc, the tail of
which is in the corresponding ``consider'' step. For instance
\begin{equation}
\begin{tabular}{c}
\miz{consider} x \miz{be} \miz{set} \miz{such} \miz{that}\\
$\strut\mkern20mu$Ai\miz{:} $\phi$\miz{(}\mbox{x}\miz{)} \miz{by} Aj \miz{;}\\
\end{tabular}\label{step_consider_example}
\end{equation}
is a tail of a non-reference arc with the head that contains the statement
\begin{equation}
x=x
\end{equation}
in a step of the reasoning located after the step \eqref{step_consider_example}. Additionally, the {\it verum} in the Mizar
language is stated as
\miz{not} \miz{contradiction}.

\newcommand{\Just}[1]{\mbox{Justification}(#1)}
\newcommand{\Var}[0]{\mathcal{V}ar}
\newcommand{\Stat}[1]{\mbox{Statement}(#1)}
\newcommand{\Step}[1]{\mbox{Step}(#1)}
\begin{thm}\label{constr}
Let $G=\langle V,E\rangle$ be a \DAG{} and $E_1$ be a subset of $E$. Suppose that there exists a vertex which either as a head
or as a tail belongs to arcs of $E_1$ only. Then there exists
 a correct Mizar reasoning that
 has the structure described by $G$ and the set of reference arcs equal to $E_1$.
\end{thm}
\begin{proof}
Let $G$, $E_1$ fulfil assumptions and let $\mathfrak{t}\in V$ satisfy $\{\langle v,u\rangle \in E: v=\mathfrak{t}\vee
u=\mathfrak{t}\}\subseteq E_1$. Let us fix a topological sorting $\tau\in \TS(G)$. We introduce the following notation:
$\Var\coloneq\{v\in V:\underset{u\in V}\exists\langle v,u\rangle \in E\setminus E_1\}$ and
\begin{equation}
\begin{array}{@{}c}
\Just{v}
\coloneq\left\{
\begin{array}{@{}l@{\quad}p{160pt}@{}}
\miz{by}\:\mbox{A}\tau(u_1)\!\miz{,}\,\mbox{A}
\tau(u_2)\!\miz{,}\ldots\miz{,}\mbox{A}\tau(u_k) \miz{;}
& if $\:N^{-}_{\langle V,E_1\rangle}(v) =\{u_1,u_2,\ldots,u_k\}$ $\wedge$ $\strut\mkern22mu k>0$,\\
\miz{;}
& if $\: N^{-}_{\langle V,E_1\rangle}(v)=\emptyset$,
\end{array}
\right.
\\
\begin{array}{@{}c}
\Stat{v}\coloneq\left\{
\begin{array}{@{}p{150pt}@{\quad}p{160pt}@{}}

$\mbox{x}\tau(u_1)\miz{=}\mbox{x}\tau(u_1)\:
\miz{\&}\:\mbox{x}\tau(u_2)\miz{=}\mbox{x}\tau(u_2)$
$\strut\mkern85mu\miz{\&}\ldots
\miz{\&}\:\mbox{x}\tau(u_k)\miz{=}\mbox{x}\tau(u_k)$
& if $\:N^{-}_{\langle V,E_1\rangle}(v) =\{u_1,u_2,\ldots,u_k\}$ $\wedge$ $\strut\mkern22mu k>0$,\\
 \miz{not}\: \miz{contradiction} & \mbox{if}\;
 $N^{-}_{\langle V,E\setminus E_1\rangle}(v)=\emptyset,$
\end{array}
\right.
\end{array}
\end{array}
\end{equation}
where $v\in V$. The reasoning step corresponding to a vertex $v$, denoted by $\Step{v}$, is defined as follows:
\begin{equation}
\Step{v}\coloneq\left\{\begin{array}{@{}l@{\quad}p{60pt}@{}}
\miz{consider}\:\mbox{x}\tau(v)\:\miz{be}\:\miz{set}\:\miz{such}\: \miz{that}
&if $v\in\Var$,\\
\strut\mkern20mu\mbox{A}\tau(v)\,\miz{:}
\mbox{x}\tau(v)\,\miz{=}\,\tau(v)\,\miz{\&}\,\Stat{v}\:\Just{v}&\\
\mbox{A}\tau(v)\,\miz{:}\,\Stat{v}\:\Just{v}&
if $v\not\in\Var$.
\end{array}\right.
\end{equation}

\nopar If we prove that the reasoning
{\small
\begin{equation}
\begin{array}{l}
\miz{not}\:\miz{contradiction}\\
\strut\mkern20mu\miz{proof}\\
\strut\mkern40mu\Step{\tau^{-1}(1)}\\
\strut\mkern40mu\vdots\\
\strut\mkern40mu\Step{\tau^{-1}(\tau(\mathfrak{t}-1)}\\
\strut\mkern40mu\miz{thus}\:\Step{\mathfrak{t}}\\
\strut\mkern40mu\Step{\tau^{-1}(\tau(\mathfrak{t}+1)}\\
\strut\mkern40mu\vdots\\
\strut\mkern40mu\Step{\tau^{-1}(|V|)}\\
\strut\mkern20mu\miz{end;}\\
\end{array}
\end{equation}
}
is correct and has expected properties, the proof will be completed. For illustration, an example of a reasoning that follows
this pattern is presented in Fig.~\ref{Other}. Naturally, justifications in none of the steps are necessary for the Mizar
system (an empty ``semicolon'' justification suffices to have it accepted by the checker), but a proof structure based on these
justifications and statements has properties that are expected from the family of reference arcs and non-reference ones.
Additionally, $\{u\!\in \!V\!:\!\langle u,v\rangle\! \in\! E\setminus E_1\}\subseteq\Var$ for any $v\in V$. Moreover, none of
variable identifiers are used before their introduction in reasoning. Analogously none of label identifiers are used before
their introduction in reasoning. It is also easy
to check that if $\Stat{\mathfrak{t}}=\miz{not}\:\miz{contradiction}$ then statement of the reasoning step
$\miz{thus}\:\Step{\mathfrak{t}}$ is equal to the final goal, hence finally the proof is completed.
\end{proof}
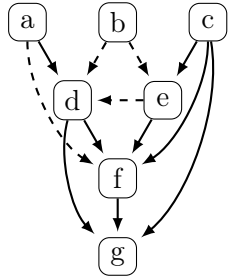
\begin{figure}[h!]
\begin{center}
$\strut\mkern-40mu$%
\begin{tikzpicture}
[>=stealth,->,shorten >=2pt,looseness=.5,scale=1.2];
\tikzstyle{every node}=[myrectanglenodes];
      \path (0:0cm) node (f) {f};
\path(60:1cm) node (e) {e};
\path (120:1cm) node (d) {d};
\path (60:2cm) node (c) {c};
\path (120:2cm) node (a) {a};
\path (90:1.73cm) node (b) {b};
\path (270:0.866cm) node (g) {g};
\draw[Referencja](a)--(d);
\draw[Referencja](c)--(e);
\draw[Referencja,bend left, looseness=0.8](c)to(f);
\draw[Referencja,bend left, looseness=0.8](c)to(g);
\draw[Referencja](e)--(f);
\draw[Referencja](d)--(f);
\draw[Referencja](f)--(g);
\draw[Referencja,bend right, looseness=0.8](d)to(g);
\draw[Variable,bend right, looseness=0.8](a)to(f);
\draw[Variable](e)--(d);
\draw[Variable](b)--(d);
\draw[Variable](b)--(e);
 \draw[xshift=2.4cm,yshift=0.4cm] node [right,text width=6cm,draw=none]
    {\small
    $\begin{array}{l}
\miz{not}\:\miz{contradiction}\\
\strut\mkern20mu\miz{proof}\\
\strut\mkern40mu\miz{consider}\:\mbox{x1}\:\miz{be}\:\miz{set}\:\miz{such}\: \miz{that}\\
\strut\mkern60mu\mbox{A1}\miz{:}\,
\mbox{x1}\,\miz{=}\,1\,\miz{\&}\,\miz{not}\:\miz{contradiction}\miz{;}\\
\strut\mkern40mu\miz{consider}\:\mbox{x2}\:\miz{be}\:\miz{set}\:\miz{such}\: \miz{that}\\
\strut\mkern60mu\mbox{A2}\miz{:}\,\mbox{x2}\,\miz{=}\,2\,\miz{\&}\,\miz{not}\:\miz{contradiction}\miz{;}\\
\strut\mkern40mu\mbox{A3}\miz{:}\,\miz{not}\:\miz{contradiction}\miz{;}\\
\strut\mkern40mu\mbox{A4}\miz{:}\,\mbox{x2}\,\miz{=}\,\mbox{x2}\:\miz{by}\:\mbox{A1}\miz{;}\\
\strut\mkern40mu\miz{consider}\:\mbox{x5}\:\miz{be}\:\miz{set}\:\miz{such}\: \miz{that}\\
\strut\mkern60mu\mbox{A5}\miz{:}\,
\mbox{x5}\,\miz{=}\,5\,\miz{\&}\,\mbox{x2}\,\miz{=}\,\mbox{x2}\:\miz{by}\:\mbox{A3}\miz{;}\\
\strut\mkern40mu\mbox{A6}\miz{:}\,\mbox{x1}\,\miz{=}\,\mbox{x1}\:\miz{by}\:\mbox{A3}\miz{,}\,\mbox{A4}\miz{,}\,\mbox{A5}\miz{;}\\
\strut\mkern40mu\miz{thus}\:\mbox{A7}\miz{:}\:\miz{not}\:\miz{contradiction}\:\miz{by}\:\mbox{A3}\miz{,}\,\mbox{A4}\miz{,}\,\mbox{A6}\miz{;}\\
\strut\mkern20mu\miz{end;}\\
\end{array}$};
\end{tikzpicture}
\end{center}
\caption{The Mizar proof script that illustrates the construction from the proof of Theorem~\ref{constr}, the solid arrows represent elements of $E_1$, dashed arrows represent elements of $E\setminus E_1$, $\mathfrak{t}=g$,
$\Var=\{a,b,e\}$ and $\tau(a)=1,\, \tau(b)=2,\, \ldots\,,\,\tau(g)=7$.}
\end{figure}
\pagebreak

To be precise, the Mizar system will not accept many of the reasonings given by the construction in the proof. The reason is that
the real system has a restriction on the number of references that can be used in the justification of a single step. In the
{official Mizar distribution} this limit is 25, and this corresponds to the limit for the in-degree of every vertex in abstract
proof graphs. This restriction is only used to bound the time of correctness verification of a single step and so far there was
no need to increase this limit.

It is easy to see that the construction above can also serve to prove a parametrised version of Theorem~\ref{constr} that takes
into account this parameter.

\begin{thm}\label{constr_1}
Let $G=\langle V,E\rangle$ be a \DAG{} with in-degree bounded by $k$ and $E_1$ be a subset of $E$. Suppose that there exists a
vertex which either as a head or as a tail belongs to arcs of $E_1$ only. Then there exists
 a correct Mizar reasoning with the number of references that
can be used in the justification bounded by $k$ and that
 has the structure described by $G$ as well as the set of reference arcs equal to $E_1$.
\end{thm}

Note that both legibility improvement methods, 1st MIL and 2nd MIL, are NP-complete even with the limit 2 for the in-degree
(see Th.~\ref{APHlimit2}).

\section{Other methods of improving proof readability}\label{Other}

\par Various methods of improving readability of natural deduction proofs which use ``then'' $\tau$-steps
have been considered in~\cite{KPslgr}. Now we sketch the complexity of optimization problems corresponding to the other proof
readability criteria. It is not the main purpose of this paper, but it additionally illustrates the {difficulty} of improving
the readability {through} the modification of  the order of independent steps written in the proof script. These methods focus
on the location of information used to justify a step and ones that are most similar to the methods considered in the current
paper are:
\begin{enumerate}[label=(\roman*)]
\item minimization of
the largest distance between a premise and its use measured  as the number of steps between these two steps in a linear
representation of a proof script,
\item or minimization of the sum of all distances between every premise and its use.
\end{enumerate}
These methods of improving legibility of proofs can be formulated as the following two problems:\\

\nopar{\bf 3rd Method of Improving Legibility} ({\bf$3$rd MIL}):\\
$\strut\mkern25mu$\textsc{Instance}: A \DAG\: $G = \langle V,E\rangle$, a subset $E_1$ of $E$, and a positive integer $K\leq
|V|$.\\ $\strut\mkern25mu$\textsc{Question}: Does there exist a topological sorting $\tau$ of $G$ for which $$\underset{\langle
v,u\rangle\in E_1}{\forall}\tau(u)-\tau(v) \leq K\mbox{?}$$\\

\nopar{\bf 4th Method of Improving Legibility} ({\bf$4$th MIL}):\\
$\strut\mkern25mu$\textsc{Instance}: A \DAG\: $G = \langle V,E\rangle$, a subset $E_1$ of $E$,  a positive integer $K\leq
\left(\begin{array}{c}\!|V|+1\!\\3\end{array}\!
    \right)$.\\
$\strut\mkern25mu$\textsc{Question}: Does there exist a topological sorting $\tau$ of $G$ for which $$\underset{\langle
v,u\rangle\in E_1}{\sum}\tau(u)-\tau(v) \leq K\mbox{?}$$\\

\noindent Obviously the expression $\tau(u)-\tau(v)$ describes the distance between the pair of vertices $v$, $u$ and the arc $\langle
v,u\rangle\in E_1$ describes the situation where the statement formulated in a reasoning step corresponding to the vertex $v$
is used as a premise in the justification of a step that corresponds to the vertex $u$.

Since every abstract proof graph $G=\langle V,E\rangle$ is constructible in the case $E=E_1$ (see Th.~\ref{constr}), it is easy
to see that problems \mbox{$3$rd} and \mbox{$4$th MIL} generalize known NP-complete problems as Directed Bandwidth and Directed
Optimal Linear Arrangement (see \mbox{GT41}, \mbox{GT43} in \cite{NPLIST}). Thus, we immediately conclude that these two
problems belong to the class of NP-complete problems, too.

Another method of improving  readability of the proofs considered in~\cite{KPslgr} concerns the number of steps that have to be
marked with respective labels. As it was mentioned in Section~\ref{Graph_representation_of_proofs}, labels are used to pass on
the information contained in a step to justify another one. Obviously, if some label is used only once in the following step
and can be replaced by the ``then'' construction then this label can be removed. Therefore, it is possible to minimize the
number of labels that are introduced in the proof script. Minimization of this parameter can be formulated as follows:\\

\nopar{\bf 5th Method of Improving Legibility} ({\bf$5$th MIL}):\\
$\strut\mkern25mu$\textsc{Instance}: A \DAG\: $G = \langle V,E\rangle$, a subset $E_1$ of $E$, and a positive integer $K\leq
|V|$.\\ $\strut\mkern25mu$\textsc{Question}: Does there exist a topological sorting $\tau$ of $G$ for which $$\{v\in
V:\underset{u\in V}\exists \langle v,u\rangle \in E_1\wedge\tau(u)-\tau(v) >1 \}$$ has size at most $K$?\\

We show in Theorem~\ref{MIL_5} below that the problem \mbox{$5$th MIL} is NP-complete. However, the subfamily of abstract proof
graphs considered in the proof of the theorem cannot be realized by proof scripts in
 MML due to an additional syntax restriction of Mizar.
In the Mizar system a step of reasoning which introduces variables satisfying a statement which is used further in the script
(see \eqref{step_consider_example}, the variable \mbox{x} satisfies $\phi$\miz{(}\mbox{x}\miz{)}) must be decorated with a
label ($Ai$). This condition can be expressed in terms of graphs as follows: every vertex that is both a tail of some reference
arc and a tail of some non-reference arc corresponds to a proof step with a label. As a result the minimalization must take the
following form.\\

\nopar{\bf 5th Method of Improving Legibility limited to the Mizar system} ({\bf$5$th MIL$_{{MIZ}}$}):\\
$\strut\mkern25mu$\textsc{Instance}: A \DAG\: $G = \langle V,E\rangle$, a subset $E_1$ of $E$, and a positive integer $K\leq
|V|$.\\ $\strut\mkern25mu$\textsc{Question}: Does there exist a topological sorting $\tau$ of $G$ for which $$\{v\in
V:\underset{u\in V}\exists \langle v,u\rangle \in E_1\wedge(\tau(u)-\tau(v) >1 \vee \underset{w\in V}\exists \langle
v,w\rangle\in  E\setminus E_1)\}$$ has size at most $K$?\\

It turns out that this limitation has a significant impact on the complexity of the task to minimise the number of labels. We
show in Theorem~\ref{MIL_5_MIZ} that the problem $5$th MIL$_{{MIZ}}$ is solvable in polynomial time. In consequence we get that
it is possible to minimize the number of labels in the Mizar proof scripts effectively.

\begin{thm}\label{MIL_5}
$5$th \mbox{MIL} in NP-complete.
\end{thm}
\begin{proof}
We transform \VS{} to $5$th \mbox{MIL}. Let an undirected graph $G=\langle V,E \rangle$ and a positive integer $K\leq |V|$ be
an instance of \VS{}. We can assume that $G$ is without self-loops, i.e.\ without edges that connects a vertex to itself. We
construct a directed graph $G^\prime = \langle V^\prime, E^\prime\rangle$ and a subset of arcs $E_1 \subset E^\prime$ such that
there exists a vertex cover of $G$ with the size at most $K$ if and only if there exists a topological sorting $\tau\in
\TS(G^\prime)$ for which $|\mathfrak{L}(\tau)|\leq K$, where
\begin{equation}\mathfrak{L}(\tau)\coloneq\{v\in V^\prime:\underset{u\in V^\prime}\exists \langle v,u\rangle \in
E_1\wedge\tau(u)-\tau(v) >1 \}.
\end{equation}
Let $G^\prime$, $E_1$ be defined by
\begin{equation}
\begin{array}{rcl}
V^\prime &=& V\times \{0,1\},\\ E^\prime&=&\{\langle \langle v,0\rangle,\langle v,1\rangle\rangle:v\in V\}\cup
\{\langle \langle v,0\rangle,\langle u,1\rangle\rangle:\{v,u\}\in E\},\\
E_1&=&\{\langle \langle v,0\rangle,\langle v,1\rangle\rangle:v\in V\}.\\
\end{array}\label{def_MIL_5}
\end{equation}
This translation can clearly be done in LOGSPACE. Notice that $G^\prime$ is acyclic since for each $v\in V$ the in-degree of
vertices $\langle v,0\rangle$ is 0 and out-degree of vertices $\langle v,1\rangle$ is 0. The main idea of the proofs uses the
fact that since an edge $\{v,u\}$ belongs to $E$, at least one of $\langle v,0\rangle$ or $\langle u,0\rangle$ has to belong to
$\mathfrak{L}(\tau)$, for every $\tau\in \TS(G^\prime)$ (see Fig.~\ref{MIL_5_proof}). This is a consequence of the simple
observation that at most one of two equalities $\tau(\langle v,0\rangle)+1 =\tau(\langle v,1\rangle)$, $\tau(\langle
u,0\rangle)+1 =\tau(\langle u,1\rangle)$ can hold for every $\tau\in \TS(G^\prime)$.
\begin{figure}
\begin{center}
\begin{tikzpicture}
\def\r{1cm}
\def\s{1.2cm}
[>=stealth,->,shorten >=2pt,looseness=.5,auto]
    \matrix (M)[matrix of math nodes,
        column sep={\s,between origins},
        row sep={\r,between origins},
        nodes={mynodes,color=black}]
{|(A00)|&|(A01)|\\ |(A10)|&|(A11)| \\ };
\node at (A00) [above left, inner sep=1mm]{$\langle v, 0 \rangle$};
\node at (A10) [above left, inner sep=1mm]{$\langle v, 1 \rangle$};
\node at (A01) [above right, inner sep=1mm]{$\langle u, 0 \rangle$};
\node at (A11) [above right, inner sep=1mm]{$\langle u, 1 \rangle$};
\draw[Referencja](A00)--(A10);
\draw[Referencja](A01)--(A11);
\draw[Referencja](A00)--(A11);
\draw[Referencja](A01)--(A10);
\end{tikzpicture}
\end{center}
\caption{The subgraph of $G^\prime$ that corresponds to an edge $\{v,u\}$, which illustrates the construction from the proof of
Theorem~\ref{MIL_5}.}
\label{MIL_5_proof}
\end{figure}

Let $\mathcal{V}$ be a vertex cover of $G$. Let us consider a partition $\pi(\mathcal{V})$ of $G^\prime$ defined by
\begin{equation}
\pi(\mathcal{V})\coloneq\{\{\langle v,0\rangle\}:v \in \mathcal{V}\}
\cup
\{\{\langle v,1\rangle\}:v \in \mathcal{V}\}
\cup
\{\{\langle v,0\rangle,\langle v,1\rangle\}:v\in V\setminus \mathcal{V}\}.
\end{equation}
Notice that the only variables in $\pidigraph{G^\prime}{\pi(\mathcal{V})}$ that have both in-degree and out-degree non-zero are
the ones in $\{\{\langle v,0\rangle,\langle v,1\rangle\}:v\in V\setminus \mathcal{V}\}$. Moreover, the vertices of
$\pidigraph{G^\prime}{\pi(\mathcal{V})}$ in this set are not connected as $\mathcal{V}$ is a vertex cover. Then it is easy to
check that $\pidigraph{G^\prime}{\pi(\mathcal{V})}$ is acyclic and, in consequence, there exists $\tau \in
\TS(\pidigraph{G^\prime}{\pi(\mathcal{V})})$. Let $\sigma:V^\prime\rightarrow|\{1,2,\ldots,V^\prime\}|$ be a function given by
$\sigma (\langle v,i\rangle) = j + \underset{R\in \pi(\mathcal{V}):\tau(R)<\tau(P)}{\sum}|R|$ where $P$  is the only one
element of $\pi(\mathcal{V})$ that contains $\langle v,i\rangle$; $j=1$ if $|P|=1$ and $j=i+1$ if $|P|=2$. It is also easy to
check that $\sigma\in \TS(G^\prime)$ and $\mathfrak{L}(\sigma)\subseteq \{\langle v,0\rangle:v \in\mathcal{V}\}$, hence finally
$|\mathfrak{L}(\sigma)|\leq |\mathcal{V}|$.

Let $\sigma$ be a topological sorting of $G^\prime$ and let $\mathcal{V}=\{v\in V:\langle v,0\rangle
\in\mathfrak{L}(\sigma)\}$. It is obvious that $|\mathcal{V}|\leq |\mathfrak{L}(\sigma)|\leq K$, hence we only need to show
that $\mathcal{V}$ is a vertex cover of $G$. Suppose, contrary to our claim, that there exists an edge $\{v,u\}\in E$ such that
$\{v,u\}\cap \mathcal{V}=\emptyset$. Then $\sigma(\langle v,1\rangle) - \sigma(\langle v,0\rangle)\leq 1$, since
 $\langle v,0\rangle$ is the tail of $\langle\langle v,0\rangle, \langle v,1\rangle\rangle\in E_1$.
In the similar way, we claim that $\sigma(\langle u,1\rangle) - \sigma(\langle u,0\rangle)\leq 1$. Since $\sigma$ is a
topological sorting of $G^\prime$ we have also that both $\sigma(\langle v,0\rangle),\sigma(\langle u,0\rangle)$ are less than
both $\sigma(\langle v,1\rangle),\sigma(\langle u,1\rangle)$ (see Fig.~\ref{MIL_5_proof}). Additionally, these inequalities are
between natural numbers, hence finally $\sigma(\langle v,0\rangle)= \sigma(\langle u,0\rangle)$,
 $\sigma(\langle v,1\rangle)=\sigma(\langle u,1\rangle)$ and $v=u$, but this contradicts our assumption that $G$ is without
 self-loops.
\end{proof}

\begin{thm}\label{MIL_5_MIZ}
The \mbox{$5$th} \mbox{MIL}$_{MIZ}$ problem is solvable in polynomial time.
\end{thm}
\begin{proof}
Let us consider an abstract proof graph $G=\langle V,E\rangle$ and an arbitrary subset $E_1$ of $E$. Note that every step
corresponding to a vertex $v\in V$ for which one of the following properties holds
\begin{enumerate}[label=(\roman*)]
\item $|\mathcal{N}^+_{\langle V,E_1\rangle}(v)|>1$ or
\item $|\mathcal{N}^+_{\langle V,E_1\rangle}(v)|=1$ and
$|\mathcal{N}^+_G(v)|>1$
\end{enumerate}
has to contain a label. Denote by $L_1$ the set of vertices that satisfy one of these conditions and let
$\mathfrak{R}_1\coloneq\{\langle v,u\rangle\in E_1:v\in L_1\}$. Naturally, at most one piece of information belonging to
$\mathcal{N}^-_{\langle V, E_1 \rangle}(v)$, can be passed on to justification of $v$ by the ``then'' construction, hence the
other premises have to have labels. Let us construct a set $L_2$ that contains these premises. For each non-empty
$\mathcal{N}^-_{\langle V, E_1 \setminus\mathfrak{R}_1\rangle}(v)$ choose its subset $A_v$ of size  $|\mathcal{N}^-_{\langle V,
E_1 \setminus\mathfrak{R}_1\rangle}(v)|-1$ and for empty $\mathcal{N}^-_{\langle V, E_1
\setminus\mathfrak{R}_1\rangle}(v)$ let $A_v=\emptyset$. Let $L_2\subseteq V\setminus L_1$ be a union of $A_v$ for all $v\in
V$. As $\mathcal{N}^-_{\langle V, E_1 \setminus\mathfrak{R}_1\rangle}(v_1)\cap
\mathcal{N}^-_{\langle V, E_1 \setminus\mathfrak{R}_1\rangle}(v_2)=\emptyset$ for different
$v_1,v_2 \in V$, the choice of $A_v$ has no impact on the size of $L_2$.

With these considerations in mind we can conclude  that every proof script has to have at least $|L_1\cup L_2|$ labels.
Obviously {sets $L_1$, $L_2$ can be determined in polynomial time.} The proof is completed by showing that there exists a proof
linearisation which uses exactly $|L_1\cup L_2|$ labels. For this purpose let us note that maximal length paths in the DAG
$R\coloneq \langle V,E_1\setminus(\mathfrak{R}_1\cup\mathfrak{R}_2)\rangle$ are node-disjoint and the graph
$\pidigraph{R}{\pi}$ is acyclic since only the last vertex in every maximal length path of $R$ can have the out-degree in $G$
greater than 1, where $\pi$ is a partition of $V$ defined by $\pi\coloneq
\{\vertexOF{P}:P\;\mbox{is}\;\mbox{a}\;\mbox{maximal}\; \mbox{length}\;\mbox{path}\;\mbox{in}\;R\}$, {and $\mathfrak{R}_2
=\{\langle v,u\rangle\in E_1 :v\in L_2\}$.} Obviously, vertices belonging to every element of $\pi$ can be ordered into a
$\tau$-reasoning path and then setting these paths in order determined by an arbitrary topological sorting of
$\pidigraph{R}{\pi}$ we get a proof linearisation.
\end{proof}

\section{Conclusion}\label{Conclusion}

In this paper we consider several problems arising in the course of improving readability {of proof scripts}. We concentrated
on two methods of  improving proof readability based on Behaghel's First Law. For these two methods we have found a common,
unexplored problem \AHP{} concerning graph partitions. This problem is similar to the problem of finding a partition of \DAG{}
into at most $K$ node-disjoint paths, which is solvable in polynomial time \cite{DAGpartition}. With this in mind we can
observe that the requirement of acyclicity imposed on the members of the partition is essential. Precisely this is the
difference here.

{Other three methods of  improving proof readability are also solved but their solution was much simpler. These three methods
generalize the previously described and solved graph problems. Additionally, many approximate solutions of these problems have
been already found that can be adapted to the proof script case. }

\par The next step in the process of improving proof readability should be finding the algorithms that approximate the problem
\AHP{} as well as the algorithms that approximate problems $1$st \mbox{MIL} and $2$nd \mbox{MIL}.
It is also important to find heuristics that quickly give satisfactory approximate solutions to all five \mbox{MIL} problems.
This should enable the possibility to make a decision which methods are more appropriate in particular cases.

\end{document}